\documentclass[10pt,journal,compsoc]{IEEEtran}

\usepackage{algorithmic}
\usepackage{cite}
\usepackage{booktabs}
\usepackage{url}
\usepackage{epsfig}
\usepackage{multirow}
\usepackage{epstopdf}
\usepackage{subfig}
\usepackage{graphicx}
\usepackage{xspace}
\usepackage{cite}
\usepackage[T1]{fontenc}
\usepackage{amsmath} 
\usepackage{amssymb} 
\usepackage{caption}

\usepackage[linesnumbered,ruled,vlined]{algorithm2e}

\usepackage{wasysym} 

\usepackage{ragged2e} 
\renewcommand{\raggedright}{\leftskip=0pt \rightskip=0pt plus 0cm}

\newtheorem{property}{Property}
\newtheorem{definition}{Definition}
\newtheorem{example}{Example}
\newtheorem{lemma}{\textsc{Lemma}{\bfseries}{\itshape}}

\newcommand{\btitle}[1]{\vspace{1.5ex}\noindent\textbf{#1}}

\newtheorem{theorem}{Theorem}

\newtheorem{proof}{Proof}[section]
\newcommand{\ie}{\emph{i.e.,}\xspace}
\newcommand{\eg}{\emph{e.g.,}\xspace}
\newcommand{\resp}{\emph{resp.,}\xspace}

\newcommand{\eat}[1]{}

\DeclareSymbolFont{ugmL}{OMX}{mdugm}{m}{n}
\SetSymbolFont{ugmL}{bold}{OMX}{mdugm}{b}{n}
\DeclareMathAccent{\wideparen}{\mathord}{ugmL}{"F3}

\def\EndOfProof{\nolinebreak\ \hfill\rule{2.0mm}{2.0mm}}
\begin{document}
\title{ONCE and ONCE+: Counting the Frequency of Time-constrained Serial Episodes in a Streaming Sequence}
\author{%
{Hui Li{\small $~^{\#1}$}, Sizhe Peng{\small $~^{\#2}$}, Jian Li{\small $~^{*3}$}, Jingjing Li{\small $~^{\dag4}$}, Jiangtao Cui{\small $~^{*5}$}, Jianfeng Ma{\small $~^{\#6}$} }%
\vspace{1.6mm}\\
\fontsize{10}{10}\selectfont\rmfamily\itshape
$^{\#}$\, School of Cyber Engineering, Xidian University, China\\
\fontsize{9}{9}\selectfont\ttfamily\upshape
%
$^{1}$\,hli@xidian.edu.cn\\
$^{2}$\,735451908@qq.com\\
$^{6}$\,jfma@mail.xidian.edu.cn%
\vspace{1.2mm}\\
\fontsize{10}{10}\selectfont\rmfamily\itshape
$^{*}$\,School of Computer Science and Technology, Xidian University, China\\
\fontsize{9}{9}\selectfont\ttfamily\upshape
$^{3}$\,lijian021@gmail.com\\
$^{5}$\,cuijt@xidian.edu.cn
\vspace{1.2mm}\\
\fontsize{10}{10}\selectfont\rmfamily\itshape
$^{\dag}$\,Department of Computer Science and Engineering, The Chinese University of Hong Kong, China\\
\fontsize{9}{9}\selectfont\ttfamily\upshape
$^{4}$\,lijj@cse.cuhk.edu.hk
}

\IEEEtitleabstractindextext{%
\begin{abstract}
As a representative sequential pattern mining problem, counting the frequency of serial episodes from a streaming sequence has drawn continuous attention in academia due to its wide application in practice, e.g., telecommunication alarms, stock market, transaction logs, bioinformatics, etc. Although a number of serial episodes mining algorithms have been developed recently, most of them are neither stream-oriented, as they require multi-pass of dataset, nor time-aware, as they fail to take into account the time constraint of serial episodes. In this paper, we propose two novel one-pass algorithms, ONCE and ONCE+, each of which can respectively compute two popular frequencies of given episodes satisfying predefined time-constraint as signals in a stream arrives one-after-another. ONCE is only used for non-overlapped frequency where the occurrences of a serial episode in sequence are not intersected. ONCE+ is designed for the distinct frequency where the occurrences of a serial episode do not share any event. Theoretical study proves that our algorithm can correctly mine the frequency of target time constraint serial episodes in a given stream. Experimental study over both real-world and synthetic datasets demonstrates that the proposed algorithm can work, with little time and space, in signal-intensive streams where millions of signals arrive within a single second. Moreover, the algorithm has been applied in a real stream processing system, where the efficacy and efficiency of this work is tested in practical applications.
\end{abstract}

}

\maketitle

\IEEEdisplaynontitleabstractindextext

%
\IEEEpeerreviewmaketitle
\section{Introduction}
\label{intro}

With the development of cloud computing, internet of things, biocomputing and so on, numerous ordered sequences are accessible from various daily applications. Among all these applications, mining serial episodes from long sequences has various potential applications and thereby drew much research attention, especially in the fields of telecommunication~\cite{DBLP:journals/datamine/MannilaTV97}, finance~\cite{DBLP:conf/eisic/GolmohammadiZ12}, neuroscience~\cite{DBLP:journals/dke/AcharAS13} and information security. A \emph{serial episode} is referred to as an ordered collection of specific signals, \eg sequential alarm pattern in telecommunication alarm sequence, and the times it appears in the sequence is referred to as its frequency. Generally, studying the frequency of serial episode patterns can be used to analyze or summarize the whole sequence and can also be used to predict future signals in the sequence.

For instance, a long telecommunication alarm sequence can be summarized using a limited number of representative serial episodes; on the other hand, we may be interested in the frequency of some specific alarm episode patterns within the long sequence so that responses towards these alarm episodes can be optimized; besides, it can be directly used to mine frequent serial episodes; we may also be interested in predicting the future alarms within the sequence. All of these examples require counting the frequencies of a given set of serial episodes. \eat{However, in the process of counting the frequency of sequential patterns, techniques utilized in current works shows inability when applied to dynamic data streams }

Counting the frequency for a finite set of given serial episodes can be easily found in many real applications in different fields. For instance, in securities market, the detection of securities fraud is a challenging task considering the massive amount of trading data produced everyday. Insider trading, one category of deceptive practices, can be generalized as a serial pattern using a group of actions including offers and sales of securities~\cite{DBLP:conf/eisic/GolmohammadiZ12}. With a set of patterns/trends that is known to be fraud, automatic detection of fraudulent activities can be achieved as long as we focus on the deceptive patterns in the streaming trading sequence. Besides, in the field of bioinformatics, in order to analyze a gene set of interest, analyzing its frequency and distribution among the whole genome datasets can help find out in which tissues or cells are they co-expressed~\cite{sahaweb}.

In a message-intensive system, millions of data are generated within several minutes. Such data are referred to as streams~\cite{DBLP:journals/corr/abs-1205-4477}. Formally, a \emph{streaming sequence} is composed of several types of events and it will dynamically update its length as new events occur and often in a high rate. Conventional methods for counting the frequency of serial episodes are generally based on the idea of storing the entire dataset and then processing it through multiple passes. Hence, traditional algorithms are not applicable on streams as it is impossible to store the entire unlimited data before the processing. Any method for data streams must thus operate under the constraints of limited memory and time which means that data streams must be processed faster than they are generated. To this end, in this paper, we propose an efficient one-pass solution to count the frequencies of a given set of serial episodes in a data stream without the need to load the whole sequence beforehand.

In addition, although there exist some efforts that mine the frequency of serial episodes from a long sequence, among which~\cite{DBLP:conf/kdd/LaxmanSU07,DBLP:conf/sigmod/WuDR06} even work in streams in a one-pass manner, they suffer from a key limitation that the serial episodes mined are not associated with any time constraint. That is, they do not care whether the serial episodes fall into a limited time span (\eg an hour, a day, etc.). For instance, a common scenario in telecommunication alarm sequence study is to learn the typical serial alarm episodes in order to discover the sequential association rules between alarms so that we do not need to respectively respond to each of them, because responding to the earliest alarm can always automatically address the following ones incurred by it. Obviously, alarms that form a sequential association rule should not exhibit too large time span (\eg an hour, etc.). As another example, we may be interested to know a particular person's \emph{daily} mobility pattern (e.g., \emph{Office}$\rightarrow$\emph{Gym}$\rightarrow$\emph{Bar}) to help quantify his daily movement condition. In both of these scenarios, we have to limit the time span of the serial episodes.

As the state-of-the-art single-pass serial episodes mining algorithms,~\cite{DBLP:conf/kdd/LaxmanSU07,DBLP:conf/sigmod/WuDR06,DBLP:conf/edbt/CadonnaGB11} employ automata to count the occurrences for each target episode. Unfortunately, the automata they employed cannot be easily incorporated with time constraint. This is thoughtfully discussed in Section~\ref{sec:probstat}. To address the problem, we propose a new model that successfully avoids the problem of~\cite{DBLP:conf/kdd/LaxmanSU07,DBLP:conf/sigmod/WuDR06,DBLP:conf/edbt/CadonnaGB11} in counting serial episodes satisfying given time span within \emph{streaming sequence}. In summary, our contributions in this work are as follows.
\begin{itemize}
\item We formally define the non-overlapped frequency counting problem of time-constrained episodes. To address the problem, we present a carefully designed data structure, namely OccMap, as well as a group of operations over it. An OccMap corresponds to a particular serial episode and stores the timestamps of valid signals
  \item Based on OccMap,we propose two efficient algorithm ONCE and ONCE+ (\underline{O}ccurre\underline{N}ce \underline{C}ount of serial \underline{E}pisode) to compute two popular frequencies of given time-constrained serial episodes in a dynamic event stream, over which only one-pass process is required. In particular, ONCE computes non-overlapped frequency while ONCE+ works on distinct frequency.
  \item ONCE (ONCE+) does not require any other user-specified parameter except the time constraint $\tau$. Our algorithm does not put any restriction over the streaming sequence, which can either arrive in batches (a group of sequentially ordered signals)~\cite{DBLP:conf/icdm/PatnaikLCR12} or single signals.
  \item We theoretically prove that ONCE and ONCE+ algorithms can correctly count the target frequencies, respectively. Besides, processing an event in the stream only requires $O(k\log\tau)$ time, where $k$ is the length of the episode.
   \item Empirical studies conducted over both real-world and synthetic datasets justify that ONCE and ONCE+ can efficiently and correctly find the frequencies of the serial episodes and outperforms baseline method in the aspects of both space and time cost.
\end{itemize}

The rest of this paper is organized as follows. In next section, we briefly discuss related work in serial episode mining. In Section~\ref{sec:probstat} we introduce the preliminary definitions and problem statement. Afterwards, we present the details of our solution towards the problem in Section~\ref{sec:meth} and Section~\ref{sec:once+} with theoretical study of the complexity and correctness. In Section~\ref{sec:exp}, we conduct empirical study over real-world and synthetic datasets. We show a practical application where the proposed algorithm is applied and discuss the corresponding observations in Section~\ref{sec:obs}. Lastly, we conclude our work in the Section~\ref{sec:con}.
\vspace{0ex}\section{Related Work}\label{sec:rel}
Several types of sequential patterns have been extensively studied so far, including frequent (closed) sequential pattern mining~\cite{DBLP:conf/icde/AgrawalS95,DBLP:conf/icde/PeiHPCDH01,DBLP:conf/sdm/YanHA03,DBLP:journals/csur/MooneyR13,DBLP:conf/kdd/BertensVS16}, serial episodes discovery~\cite{DBLP:journals/kais/JiBD07,DBLP:conf/kdd/LaxmanSU07}, periodic (ordered) pattern mining~\cite{DBLP:journals/tkdd/ZhangKCY07,DBLP:journals/tkde/PeiWLWWY06}. Within these works, various frequency definitions of episodes have been proposed, which have given rise to different types of frequent episodes. Recently, Achar et al.~\cite{DBLP:journals/kais/AcharLS12} reviewed 7 different frequency definitions in the literature. Three of them, window-based frequency~\cite{DBLP:journals/datamine/MannilaTV97}, head frequency~\cite{chen9}, and total frequency~\cite{chen9}, consider the number of windows containing at least one occurrence of an episode, where each window has the same specified width. The remaining definitions, minimal occurrence-based frequency~\cite{DBLP:journals/datamine/MannilaTV97}, non-overlapped frequency~\cite{DBLP:journals/tkde/LaxmanSU05}, non-interleaved frequency~\cite{chen12} and distinct frequency~\cite{chen10}, directly take into account the different occurrences of an episode in the sequence.

However, these efforts cannot be deployed to some real-world applications such as fraud trading detection or telecommunication alarm responses as they ignored the practical significance of time constraint in episodes. Notably, the author in~\cite{DBLP:conf/kdd/LaxmanSU07} suggested that, by attaching to each automaton a time constraint, their method can address time-constrained serial episode mining problem. However, they actually failed to empirically test this suggested method in time-constrained problem. Unfortunately, as we will illustrate in detail in Section~\ref{ssec:prodef}, this suggested method is unable to generate correct answer.

In the field of serial episode mining over sequential streams, related algorithm studies have become increasingly prevalent over the recent years~\cite{DBLP:conf/icdm/PatnaikLCR12,DBLP:conf/vldb/2002,DBLP:conf/icdm/2007,DBLP:journals/kais/ChengKN08}. Patnaik et al.~\cite{DBLP:conf/icdm/PatnaikLCR12} considered serial episode mining over dynamic data streams. The main contribution of their work is to define the batch of events and apply their algorithms over each batch. But the performance of their method highly depends on the size of batches where the frequency is computed. A large batch leads to high response time, while a small one fails to count the frequency of long episodes. Especially, once each batch of data contains only one event, \ie events arrive one after another, their algorithm cannot work anymore. In addition, when a serial episode stretches over two consecutive batches, this occurrence of the episode will be missed. Xiang et al.~\cite{DBLP:conf/icde/AoLLZH15} presented MESELO algorithm, which requires a complete view of the whole sequence. It strictly limits their application in streams where the number of events is potentially unlimited. For instance, if we want to learn the frequency of an episode in the past 48 hours, the window size $\Delta$ in their method should be set as a large time span to store all records in the past hours, which takes enormous memory consumption. They also presented another work~\cite{DBLP:conf/icde/AoLWZH17} that aims to mine serial episodes over precise-positioning sequences, where the elapsed time between any two consecutive events is a constant. SASE~\cite{DBLP:conf/sigmod/WuDR06} has been proposed to record the appearance of target serial episode within a stream. The proposed structure has to spend $O(2k\tau)$ time to process a single signal in the stream, while our algorithm takes only $O(k\log\tau)$. Besides, none of these works takes into account the time span for the target serial episode. In contrast, we present in this paper a novel one-pass algorithm that works on stream sequence without any requirement to store the whole sequence beforehand or any limitation on the batch size, while taking into account the time span of the episodes.

\vspace{0ex}\section{Problem Formulation}\label{sec:probstat}
In this section, we shall first present serials of preliminary definitions. Besides, for ease of understanding, in Table~\ref{tab:notation} we summarize the key notations that will be used in this paper.
\begin{table}[t]
\caption{Notations}
\vspace{0ex}
\begin{center}
\begin{tabular}{r|l}
\hline
\textbf{Symbols} & \textbf{Descriptions}\\
\hline
$(s_i,t_i)$ & temporal event $s_i$ happens at $t_i$\\
$S=\langle(s_1,t_1),\ldots\rangle$ & streaming sequence\\
$e=\langle \phi_1,\ldots,\phi_k\rangle$ & serial episode\\
$Occ(e,S)$ & occurrence of $e$ in $S$\\
$Occ_{OPT}(e,S)$ & minimal occurrence of $e$ in $S$\\
$e^\tau$ & time-constrained serial episode\\
$\tau$ & time constraint of $e^\tau$\\
$k$ & length of serial episode \\
$OM(e^\tau)=[L_1,\ldots,L_k]$ & OccMap of $e^\tau$\\
$L_i$ & timestamp list of $i$-th layer in $OM(e^\tau)$\\
\hline
\end{tabular}
\end{center}
\label{tab:notation}
\end{table}

\vspace{0ex}\subsection{Preliminaries}\label{ssec:31}
We first define streaming sequences, serial episodes~\cite{DBLP:conf/kdd/WuLYT13} and non-overlapped frequency.

\begin{definition}[Streaming sequence]\label{df:streamseq} \emph{Streaming sequence} is a long (potentially infinite) sequence of event\footnote{\scriptsize To avoid duplicate word usage, we shall use the words event and signal interchangeably in the rest of this paper.}. Let $\Sigma$ be finite alphabet set, $S$ be a sequential list of events, denoted by $S=\langle (s_1,t_1), \ldots, (s_n,t_n),\ldots\rangle,$ where $s_i\in \Sigma$ and the pair $(s_i,t_i)$ ($1\le i\le n, t_{i-1}<t_i$) means event $s_i$ happens at timestamp $t_i$. We denote by $S(n)=\langle (s_1,t_1), \ldots, (s_n,t_n)\rangle$ as the first $n$ event subsequence of $S$. Let $S[i]$ be the $i$-th element of $S$ (i.e. $(s_i,t_i)$), $S[i].e$ and $S[i].t$ be the $i$-th event and corresponding timestamp of $S$, respectively.
\EndOfProof\end{definition}

For instance, a daily trajectory of a person can be denoted as
$S=\langle(\mathit{Home},8:00), (\mathit{Office},10:00), (\mathit{Gym},15:00), (\mathit{Bar},20:00)\rangle$ where $|S|=4$, $\Sigma=\mathit{Home, Office, Gym, Bar}$, $S[2]=(\mathit{Office},10:00)$, $S[2].e=\mathit{Office}$, $S[2].t=10:00$. In particular, if $s_i$ ($1\le i\le n$) is a set of events that happen simultaneously (\ie $s_i\subset \Sigma$), the sequence is referred to as \emph{complex streaming sequence}. Otherwise, if $s_i$ ($\forall i, 1\le i\le n$) is an individual event, it is a \emph{simple streaming sequence}.

\begin{definition}[Serial episode]\label{df:simpleEp}
A \emph{serial episode} is a set of totally ordered events, denoted by $e = \langle\phi_1,\ldots,\phi_k\rangle,$ where $\phi_i$ appears before $\phi_j$, if and only if $1\le i\le j\le k$. In particular, we denote by $|e|=k$ as the length of $e$.
\EndOfProof\end{definition}

For instance, in the above sequence example where $S=\langle(\mathit{Home},8:00), (\mathit{Office},10:00), (\mathit{Gym},15:00), (\mathit{Bar},20:00)\rangle$, $e_1=\langle \mathit{Home, Office}\rangle$ and $e_2=\langle \mathit{Home, Gym, Bar}\rangle$ are both serial episodes; the length of $e_1$ (\ie $|e_1|$) is 2 and that of $e_2$ (\ie $|e_2|$) is 3, respectively.

\begin{definition}[Occurrence]\label{df:occur}
Given a serial episode $e=\langle\phi_1,\ldots,\phi_k\rangle$, the timestamp $\langle t_1,\ldots,t_k\rangle$ is defined as the \emph{occurrence} of $e$ if $\phi_i$ happens at timestamp $t_i$. We denote by $Occ(e,S)$ as an occurrence of serial episode $e$ in $S$.
\EndOfProof\end{definition}
\vspace{0ex}
\begin{definition}[Minimal occurrence]\label{df:minocc}
Given a serial episode $e=\langle\phi_1,\ldots,\phi_k\rangle$, and its occurrence $Occ(e,S)$, namely $\langle t_1,\ldots,t_k\rangle$. If there is no other occurrence of $e$, say $\langle t_1^\prime,\ldots,t_k^\prime\rangle$, such that $t_1^\prime\ge t_1$ and $t_k^\prime\le t_k$, then $Occ(e,S)$ is called a minimal occurrence of $e$ in $S$, denoted as $Occ_{OPT}(e,S)$.
\EndOfProof\end{definition}

\begin{definition}[Time-constrained serial episode]\label{df:tepi}
A serial episode with time constraint $\tau$ is denoted as $e^\tau=\langle\phi_1, \phi_2, \ldots, \phi_k\rangle,$ where the occurrence of $e^\tau$ fall in a specified time period $\tau$ (\eg daily/weekly/monthly), that is, $|[t_1,t_k]|\le\tau$ (\ie $t_k-t_1\le\tau$). Usually, $e$ is used to represent a certain serial episode without time-constraint.
\EndOfProof\end{definition}
\vspace{0ex}
\begin{example}
Given the following sequences, $$S_1= \langle(A,1), (B,2), (A,3), (B,4), (C,6), (B,8)\rangle,$$ $$S_2= \langle(B,1), (B,2), (A,3), (B,4), (A,5), (C,8)\rangle,$$ $$S_3= \langle(B,1), (A,2), (B,4), (A,5), (C,7)\rangle,$$
serial episode $e=\langle B,A,B\rangle$ and time-constrained serial episode $e^2=\langle B,A,B\rangle$, we illustrate the occurrences of both $e$ and $e^2$ in all $S_1, S_2$ and $S_3$.

According to Definition~\ref{df:occur}, we can easily obtain
$Occ_1(e,S_1)=\langle 2,3,4\rangle$,
$Occ_2(e,S_1)=\langle 2,3,8\rangle$,
where $Occ_{OPT}(e,S_1)=Occ_1(e,S_1)$ is a minimal occurrence.

Similarly, according to Definition~\ref{df:occur}, it is easy to find that
$Occ_1(e,S_2)=\langle 1,3,4\rangle$,
$Occ_2(e,S_2)=\langle 2,3,4\rangle$,
$Occ_1(e,S_3)=\langle 1,2,4\rangle$.
Moreover, the occurrences of $e^2$ in all the above sequences are as follows,
$Occ(e^2,S_1)=\langle 2,3,4\rangle$,
$Occ(e^2,S_2)=\langle 2,3,4\rangle$,
$Occ(e^2,S_3)=\emptyset$.
\EndOfProof\end{example}

Note that the events constituting an occurrence of a serial episode are not required to be contiguous in the stream.

As reviewed in Section~\ref{sec:rel}, a number of different frequency definitions have been proposed to capture how often an episode occurs in an event sequence. We observe that existing frequency definitions can be grouped into two categories: definitions incurring dependent occurrences (\eg two occurrences of an episode may share common events) and definitions incurring independent occurrences. Due to space constraints, we focus this paper only on the type of frequency definitions incurring independent occurrences, which contains two frequency definitions: the non-overlapped frequency~\cite{DBLP:conf/kdd/LaxmanSU07,DBLP:journals/tkde/LaxmanSU05} and the distinct frequency~\cite{chen10}. We review the definitions of the two frequency measures as follows.

\begin{definition}[Non-overlapped frequency]\label{df:frequency}
In an event stream $S$, two occurrences of $e$ (\resp $e^\tau$), \ie $\langle t_1,\ldots,t_k\rangle$ and $\langle t_1^\prime,\ldots,t_k^\prime\rangle$, are non-overlapped if either $t_1^\prime>t_k$ or $t_1>t_k^\prime$. The \emph{non-overlapped frequency} of $e$ (\resp $e^\tau$) in $S$ is denoted as $freq(e,S)$ (\resp $freq(e^\tau,S)$).
 \EndOfProof\end{definition}

\begin{definition}[Distinct frequency]\label{df:dfrequency}
In an event stream $S$, two occurrences of $e=\langle \phi_1, \ldots, \phi_k\rangle$ (\resp $e^\tau$), \ie $\langle t_1,\ldots,t_k\rangle$ and $\langle t_1^\prime,\ldots,t_k^\prime\rangle$, are distinct if they do not share any event, that is $\forall 1\le i<j\le k,$ if $\phi_i=\phi_j$, $t_i\neq t_j^\prime$. The \emph{distinct frequency} of $e$ (\resp $e^\tau$) in $S$ is denoted as $freq^+(e,S)$ (\resp $freq^+(e^\tau,S)$).
 \EndOfProof\end{definition}
\vspace{0ex}
\begin{example}
Given the following sequences, $$S= \langle(A,1), (A,2), (A,3), (A,4), (B,5), (B,6)\rangle,$$
time-constrained serial episode $e^4=\langle A,A,B\rangle$, we can easily obtain
$Occ_1(e_4,S)=\langle 1,2,5 \rangle$,
$Occ_2(e_4,S)=\langle 1,3,5\rangle,\ldots,Occ_4(e_4,S)=\langle 2,3,5\rangle,\ldots,Occ_9(e_4,S)=\langle 3,4,6\rangle$,
are occurrences of $e^4$ in S, $\langle 3,4,5\rangle$ is a minimal occurrence of $e^4=\langle A,A,B\rangle$ in $S$, obviously, $\langle 3,4,6\rangle$ is another minimal occurrence. However, they overlap with each other. Thus, $freq(e^4,S)$ is $1$. On the other hand, $\langle 1,2,5 \rangle$ and $\langle 3,4,6 \rangle$ are distinct occurrences because they don't have the same timestamp $t_i$ and $1<3<5<6$. Thus, $freq^+(e^3,S)$ is 2.
\EndOfProof\end{example}




\vspace{0ex}\subsection{Problem definition}\label{ssec:prodef}

Given an event stream and the serial episode, whose frequency is to be extracted, we aim to identify the \emph{frequency of serial episodes with time constraint} from the long stream.

\begin{definition}[Time-constrained frequency counting problem]\label{def:pfemin}
Given event $S[i]$ in stream\footnote{\scriptsize It can be a simple stream sequence or a complex one, our model can work on both of them.} $S$ arrives one after another, a time-constrained serial episode $e^\tau$, \emph{time-constrained frequency counting problem} aims to evaluate $freq(e^\tau,S(i))$ whenever a new event $S[i]$ arrives.
\EndOfProof\end{definition}

The most related work with the aforementioned problem is serial episodes frequency mining in long sequences, among which the most representative is~\cite{DBLP:conf/kdd/LaxmanSU07}, an effective solution towards mining frequent serial episodes from an arbitrary long event sequence. This approach utilizes a group of automaton, each of which corresponds to a particular candidate serial episode. The approach works by sequentially scanning every event within the target sequence $S$. Each time an event is observed, the corresponding automaton who is waiting for this event (\ie the next state matches the event) is updated. Whenever an automaton comes to end state, its corresponding count increases by 1 and the automaton is reset to the start state.

\begin{figure}
\centering
  \epsfig{file=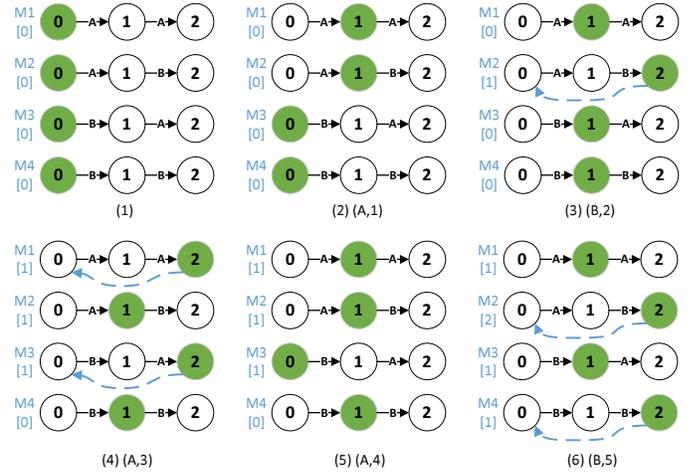, width=\columnwidth}

\vspace{0ex}\caption{State-of-the-art automaton based serial episode mining scheme [best viewed in color] (numbers in bracket are the corresponding counts).}
\vspace{0ex}\label{fig:sce}       
\end{figure}

For instance, Figure~\ref{fig:sce} shows an example where the target sequence $S=\langle(A,1), (B,2), (A,3), (A,4), (B,5)\rangle$. Suppose we are interested in the frequency of the following serial episodes, $e^1_{21}=\langle A,A\rangle, e^1_{22}=\langle A,B\rangle, e^1_{23}=\langle B,A\rangle, e^1_{24}=\langle B,B\rangle$. Given the four candidates, four finite state automata as $M1, M2, M3, M4$ are built (Figure~\ref{fig:sce}(1)) where $\Sigma=\{A,B\}$, the states of which sequentially correspond to the events in the episodes. Then it sequentially scans $S$. Event $A$ is observed first, both $M1$ and $M2$, whose next states are $A$, change to next state (Figure~\ref{fig:sce}(2)); the other two automata keep unchanged. Afterwards, $B$ is observed, this time $M2$, $M3$ and $M4$ all change to next state (Figure~\ref{fig:sce}(3)); as $M2$ reaches to the final state, hence its count increases and it is reset to the initial state again.

The aforementioned scheme~\cite{DBLP:conf/kdd/LaxmanSU07} is justified effective and efficient in mining the non-overlapped frequency of given serial episodes from a long sequence. However, it does not take into account the time constraint for each episode, hence it cannot be applied in our scenario described in Section~\ref{sec:probstat}. Can we adjust it with limited variation to address our problem? The answer is no. Although the authors in~\cite{DBLP:conf/kdd/LaxmanSU07,DBLP:journals/tkde/LaxmanSU05} suggested that simply attaching the time constraint to each automaton can solve the problem of mining episodes with given time constraint, they did not put it into practice, even when they mentioned the same method again several years later~\cite{DBLP:journals/kais/AcharLS12}. In fact, the suggested method may lead to inaccurate results in time-constrained case. The following example shows that simply adding a time constraint towards each automaton, as suggested by~\cite{DBLP:conf/kdd/LaxmanSU07,DBLP:journals/tkde/LaxmanSU05,DBLP:journals/kais/AcharLS12}, may inaccurately count the frequency of target serial episodes.

Suppose we are now adding a time constraint to each automaton in~\cite{DBLP:conf/kdd/LaxmanSU07,DBLP:conf/sigmod/WuDR06} by introducing another column, namely \emph{start\_time}, to store the timestamp of first valid state change. The count of an automaton increases when not only the state changes to the final one but also the time span between final state and \emph{start\_time} is within $\tau$.

In this way, only those instances satisfying the time constraint are counted. It seems to be a valid solution towards our problem. However, this solution may miss many valid occurrence counts, hence cannot satisfy Definition~\ref{def:pfemin}. For instance, if we follow the above adjusted solution, the count of $e^1_{21}$ will be $0$ (as $M1$ will be activated by $<(A,1),(A,3)>$ which finally fails to satisfy the time-constraint check). However, in fact there exists an instance of $e^1_{21}$, namely $\langle3,4\rangle$. Similarly, the above solution can find $1$ minimum occurrence of $e^1_{22}$ (\ie $\langle1,2\rangle$). However, there exist two minimum occurrences of $e^1_{22}$, namely $\langle1,2\rangle$ and $\langle4,5\rangle$. Such problems will be much more complex and difficult to address by automata especially when $e^\tau$ contains many repeated events, \eg $e^\tau=\langle A,A,A\rangle$.

Therefore, it is not a trivial task to design a model to count the non-overlapped frequency of serial episodes within a long streaming sequence that takes into account the time span of serial episode. To address this challenge, we develop a novel approach in next section.

\vspace{0ex}\section{ONCE Algorithm}\label{sec:meth}
In this section, we present in detail the algorithm for serial episodes counting in streaming sequence under an arbitrary time constraint. Given a streaming event sequence and a target serial episode with an arbitrary time constraint, ONCE algorithm generally works as follows. As each event in the stream passing by, we first need to find the latest minimum occurrences of the target serial episode, no matter it satisfies the given time constraint or not. To achieve that, we present a delicately constructed data structure, namely OccMap, which stores the timestamps of events that constitute the target serial episode. Whenever all events have been found in OccMap, we validate the candidate minimum occurrence by testing whether it satisfies the time constraint or not. If the test succeeds, we increase the count by 1. Afterwards, the tested occurrence and the unused timestamps of events are removed from OccMap. As a result, ONCE can output the frequency of target time-constrained serial episode whenever requested. Notably, in the following discussion, although we focus on counting the frequency of a given time-constrained serial episode, ONCE can in fact simultaneously count the frequencies for a group of target time-constrained serial episodes as the following proposed structure and corresponding operations are bind with each target time-constrained serial episode independently. Moreover, it is obvious that ONCE is a one-pass algorithm which is applicable to stream sequences.

\vspace{0ex}\subsection{OccMap structure}\label{ssec:occmap}
First of all, we present the data structure, namely OccMap, to store the timestamps of events in target serial episode. It is further used to extract candidate minimum occurrence of the serial episode.

An OccMap for time-constrained serial episode $e^\tau$, is defined as a group of hierarchical lists. In particular, given $e^\tau=\langle\phi_1,\ldots,\phi_k\rangle$, the OccMap for $e^\tau$ contains $k$ lists which are organized hierarchically into $k$ layers. The $k$ layers correspond to all the $k$ signals of $e^\tau$. Each layer is an individual list, which is used to record the timestamps of the corresponding signals in the stream. For instance, in Figure~\ref{fig:runexm}(15+) there is an OccMap that corresponds to $e^\tau=\langle A,A,B\rangle$. It consists of three lists that are hierarchically organized. Each list is correlated with a single signal in the serial episode. Notably, if the same signal appears many times in a serial episode, \ie $A$ in $e^\tau$, we assign a list to each of the appearance independently.

\begin{figure*}
\centering
  \epsfig{file=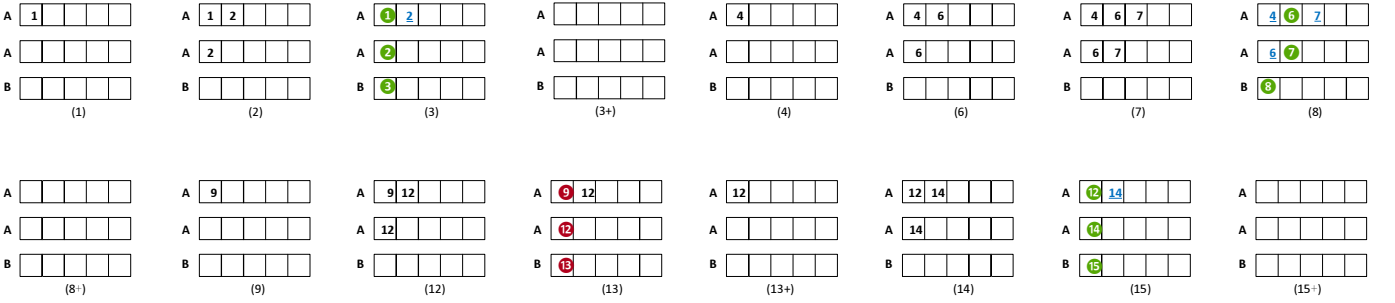, width=\linewidth}

\vspace{0ex}\caption{Mining the non-overlapped frequency of $e^3=\langle A,A,B\rangle$ in the sequence of Example~\ref{ex2} [best viewed in color]. (circled number in green: occurrences that satisfy the time constraint; circled number in red: occurrences that fail to satisfy time constraint; circled and underlined numbers: the timestamps to be removed)}
\vspace{0ex}\label{fig:runexm}       
\end{figure*}

\vspace{0ex}\subsection{Counting the frequency using OccMap}
To facilitate the following discussion, we denote by $OM(e^\tau)=[L_1,\ldots,L_k]$ as the OccMap for time-constrained serial episode $e^\tau=\langle\phi_1,\ldots,\phi_k\rangle$.

In particular, we denote by
$OM(e^\tau)[i] = L_i$, where
$L_i = \langle i_1,\ldots,i_j\rangle$ and $L_i.e = \phi_i$. $i_j$ is the timestamp of $\phi_i$ in stream $S$; $\phi_i$ refers to the signal corresponding to $L_i$.

When processing the stream, OccMap has to perform the following operations, \emph{list update}, \emph{occurrence validate} and \emph{invalid entries elimination}. In the following, we describe in detail how each of the operations are performed in OccMap. To illustrate each operation clearly, we shall use the following sequence $S$ as the running example.

\begin{example}\label{ex2} Consider the following event stream $S$:
\begin{equation*}
\begin{split}
\langle(A,1),(A,2),(B,3),(A,4),(C,5),(A,6),(A,7),(B,8),\\
(A,9),(C,10),(C,11),(A,12),(B,13),(A,14)(B,15)\rangle
\end{split}
\end{equation*}
\end{example}

\vspace{0ex}\btitle{List update.} Given an OccMap $OM(e^\tau)$ that corresponds to $e^\tau$, we perform \emph{list update} by scanning every signal in the streaming sequence $S$ as it passes by.

At the very beginning, $OM(e^\tau)$ is initialized as $k$ layered empty lists. Besides the layered lists, we denote by $OM(e^\tau).\ell$ as the most recent active layer number, which is initialized as $1$. For each signal $S[j]$ passing by, $OM(e^\tau)$ checks whether it matches any signals whose corresponding lists are active. Suppose $S[j].e=\phi_i$ and $i\le OM(e^\tau).\ell$, we append the timestamp of $S[j].e$ to the end of list $L_i$. During the update, only the layers $L_i$ (\ie $OM(e^\tau)[i]$) where $i\le OM(e^\tau).\ell$ can be updated. In another word, a new layer can be updated (\ie the corresponding list can be appended with a timestamp) only when the layers before it are not empty. It guarantees OccMap the following property, which is straightforward.
\begin{property}\label{property}[Minimum monotonicity]
If an OccMap is updated according to \emph{list update} strategy, it satisfies: $$\min L_1<\min L_2<\ldots <\min L_{OM(e^\tau).\ell-1}.$$
\end{property}

For example, given the target episode $e^3=\langle A,A,B\rangle$ and sequence $S$ shown in Example~\ref{ex2}, an OccMap $OM(e^3)$ is built, which contains three empty lists $L_1$ ($L_1.e=A$), $L_2$ ($L_2.e=A$), and $L_3$ ($L_3.e=B$). Therefore, whenever new event of $S$ arrives, only two types of events: $\langle A, t_i\rangle$ and $\langle B, t_j\rangle$ are accepted and stored in the lists. Notably, as we have mentioned above, if a signal appears many times within a target serial episode, we assign an independent list for each appearance, \ie $L_1$ and $L_2$. Suppose we are starting from the very beginning of $S$, each time a new event arrives and an old one leaves. At the very beginning, all the lists $L_1,L_2,L_3$ are initialized as empty. Besides, $OM(e^\tau).\ell$ is set to $1$. Now the first event of $S$ (\ie $(A,1)$) comes, OccMap finds that $A$ is a valid event that should be taken into account. Then it tests whether it matches any signal in the activated layers, which now contains only $L_1$. Obviously, it perfectly matches the first layer, \ie $S[1].e=L_1.e$, and $L_1$ is of cause activated, \ie $1\le OM(e^\tau).\ell$. As the test successes, we update $OM(e^3)$ by appending $S[1].t$ to the end of $L_1$, which results in Figure~\ref{fig:runexm}(1). Moreover, as the first layer is not empty then, the second layer now becomes activated, \ie $OM(e^3).\ell=2$.

When the second event in $S$ arrives, we perform the same test as above. As $S[2].e=A=L_1.e=L_2.e$ and $OM(e^3).\ell\ge 2$, both $L_1$ and $L_2$ (they are both activated) should be updated. Therefore, $S[2].t$ should be appended to the end of both $L_1$ and $L_2$, which results in the state shown in Figure~\ref{fig:runexm}(2). Similarly, as the second layer is not empty then, the third layer now becomes activated (\ie $OM(e^3).\ell=3$).

Afterward, the third event, $(B,3)$ arrives, we append it to $L_3$ and update $OM(e^3).\ell$ to 4. Now the last layer in OccMap is not empty, we have to perform another action, \emph{occurrence validation}, which is described in the following part.

\vspace{0ex}\btitle{Occurrence validation.} Whenever the bottom layer (\ie $L_k$) is updated (\ie appended by an arbitrary timestamp), it indicates that there exist some groups of timestamps in each layered list that construct a candidate minimum occurrence for the target serial episode $e^3$. Here, candidate occurrence means that it is the minimum occurrence of general serial episode without taken into account the time constraint. In fact, according to the \emph{list update} strategy, there is at least one candidate occurrence for the target serial episode once the last layer is updated.

\begin{theorem} Given that an OccMap $OM(e)$, where $|e|=k$ is updated according to \emph{list update} strategy, once $OM(e).\ell=k+1$ (\ie the last layered list is not empty), there are $\ge k$ entries in $OM(e)$, where no two entries belong to the same layered lists, that constitute an occurrence of $e$.
\end{theorem}
\begin{proof}
As the last list is not empty, $OM(e).\ell=k+1$. According to Property~\ref{property}, $\min L_1<\min L_2<\ldots<\min L_k$. If we select the entries corresponding to $\min L_1$, $\min L_2$, \ldots, $\min L_k$ from $L_1, \ldots, L_k$, respectively, the entries obviously constitute an occurrence of $e$ as they satisfy total order.
\end{proof}

In fact, there may exist many groups of entries that can constitute an occurrence of $e$. Recall that we are performing \emph{occurrence validation} once the last layered list is not empty (\ie appended by a timestamp).  In other words, $L_k$ only contains one entry, namely $L_k[1]$, when \emph{occurrence validation} is performed. That is, all the occurrences share the same end timestamp, $t_k$, at most one can affect the frequency according to Definition~\ref{df:frequency}. Therefore, we have to find the optimal occurrence, which is most probable to satisfy the time constraint, to validate. Intuitively, the optimal occurrence of $e$ should have the minimum time span (\ie $t_k-t_1$), which is in fact the \emph{minimum occurrence} shown in Definition~\ref{df:minocc}. As all the occurrences of $e$ share the same $t_k$, the minimum occurrence $Occ_{OPT}(e,S)$ that is most probable to satisfy $\tau$ should have the latest $t_1$. Therefore, to find the $Occ_{OPT}(e,S)$, we have to find the latest $t_1$ in $L_1$.

In order to find the $Occ_{OPT}(e,S)$ (\ie $[t_s,t_e]$), we traverse the OccMap in a bottom-up way. In particular, given the end stamp $t_e=L_k[1]$, we greedily find from $L_{k-1}$ a latest entry that appears before $t_e$, that is $\max L_{k-1}[j]$ subject to $L_{k-1}[j]<t_e$. Let $t_{k-1}$ be the selected entry in $L_{k-1}$, then we further greedily select from $L_{k-2}$ a latest entry that appears before $t_{k-1}$, say $t_{k-2}$. Afterwards, we iteratively perform the same selections in the upper layers, until $L_1$. In the end, we can obtain $t_1, \ldots, t_{k-1}$, which are greedily selected from $L_1, \ldots, L_{k-1}$, respectively. $t_1, \ldots, t_{k-1}, t_e$ constitute an occurrence of $e$. Obviously, $t_1$ is the latest\footnote{\scriptsize Any other $t_1^\prime >t_1$ in $L_1$ cannot be $t_s$, as it is definitely greater than $t_2$ according to our greedy selection strategy} $t_s$.

Till now, we have the minimum occurrence that is most probable to satisfy the time constraint $\tau$. Hence, we check whether its time span satisfy the time constraint by testing the inequality $t_e-t_s\le\tau$. If the test successes, we increase the frequency of the target time-constrained serial episode $e^\tau$ by $1$. Notably, if the test fails, any other occurrences will also fail, as they have smaller $t_s$ (\ie larger $t_e-t_s$).

For instance, given $e^3=\langle A,A,B\rangle$ and sequence $S$ in Example~\ref{ex2}, we have updated $OM(e)$ as $S[1], S[2]$ and $S[3]$ passed by. When $S[3]=(B,3)$ arrives, we have appended $3$ to the end of $L_3$. Once the bottom layer $L_3$ is not empty, we perform \emph{occurrence validation} as described before. In particular, we greedily find the $Occ_{OPT}(e,S)$ as $\langle1,2,3\rangle$ shown in Figure~\ref{fig:runexm}. Afterwards, we test whether its time span (\ie $3-1=2$) satisfies $\tau$ (\ie $3$). As the test successes (the corresponding entries in $OM(e)$ are circled and marked in green in Figure~\ref{fig:runexm}), we increase the frequency of $e^3$ by $1$.

\vspace{0ex}\btitle{Invalid entries elimination.} Once \emph{occurrence validation} is performed, we need to immediately eliminate invalid entries from OccMap. Depending on whether time constraint test in \emph{occurrence validation} successes or not, the elimination process varies.

If the minimum occurrence, say $Occ_{OPT}(e,S)$, found from \emph{occurrence validation} is validated as satisfying the time constraint, \ie $t_e-t_s\le\tau$, all the other entries left in $OM(e)$ are useless then. The reason is, any other occurrences that consist of any of these entries, which are smaller than $t_e$, definitely overlap with $Occ_{OPT}(e,S)$, which deviates from Definition~\ref{df:frequency}. Therefore, once the \emph{occurrence validation} successes, all the other entries in $OM(e)$ are invalid anymore, and are immediately removed from $OM(e)$. Besides, the active layer now are reset to $L_1$.

As shown in Figure~\ref{fig:runexm}(3) and (3+), except for the occurrence of the episode tested (\ie circled green entries in Figure~\ref{fig:runexm}(3)), all the other entries (\ie underlined blue entry in Figure~\ref{fig:runexm}(3)) should be eliminated, which results in Figure~\ref{fig:runexm}(3+). The same operations can be found from Figure~\ref{fig:runexm}(8) and Figure~\ref{fig:runexm}(8+), as well as Figure~\ref{fig:runexm}(15) and Figure~\ref{fig:runexm}(15+). All the lists now become empty, thus $OM(e).\ell=1$.

Otherwise, when the minimum occurrence $Occ_{OPT}(e,S)$ fails to satisfy time constraint $\tau$, we also need to find those invalid entries to eliminate from $OM(e)$. Differently, the invalid entries are no longer all the left ones in $OM(e)$ in this case. Instead, although the entries that constitute the minimum occurrence $Occ_{OPT}(e,S)$ fails to pass time constraint test, the other entries may be further used to constitute another minimum occurrence. In order to show which entries left are useless to further constitute other minimum occurrences, we present the following theory.
\begin{theorem}\label{lem:invalident}
Given that $OM(e^\tau)$ ($|e^\tau|=k$) is updated according to \emph{list update} strategy and $Occ_{OPT}(e,S)$ is found and validated by \emph{occurrence validation} process, if $Occ_{OPT}(e,S)$ which consists of $\langle t_1,\ldots,t_k\rangle$ ($t_1<\ldots<t_k$) fails to satisfy the given time constraint $\tau$ (\ie $t_k-t_1>\tau$), no other minimum occurrences $Occ_{OPT}^\prime(e,S)$ with $\langle t_1^\prime,\ldots,t_k^\prime\rangle$, where $t_k^\prime>t_k$ and $\exists i<k$ such that $t_i^\prime\le t_i$, can satisfy the time constraint either.
\end{theorem}
\begin{proof}
Suppose $Occ_{OPT}^\prime(e,S)$ that consists of $\langle t_1^\prime,\ldots,t_k^\prime\rangle$, where $t_k^\prime>t_k$ and $\exists i<k$ such that $t_i^\prime\le t_i$, satisfy the time constraint $\tau$, then
\begin{equation}\label{eq:test}
t_k^\prime-t_1^\prime\le\tau.
\end{equation}

As we are iteratively selecting the largest entry in $L_{i-1}[j]$ subject to that $L_{i-1}[j]<t_i$ according to the bottom-up minimum occurrence finding strategy presented above, we can find $L_{i-1}[j]=t_{i-1}$ and $L_{i-1}[j^\prime]=t_{i-1}^\prime$. If $t_i^\prime\le t_i$, it is straightforward to know that $j^\prime\le j$, thus $t_{i-1}^\prime\le t_{i-1}$. Similarly, we can prove that $\forall j\le i$, $t_j^\prime\le t_j$.

Therefore, $t_1^\prime\le t_1$, which means $t_k-t_1<t_k^\prime-t_1^\prime$ as $t_k^\prime>t_k$. As $t_k-t_1>\tau$, it is easy to know $t_k^\prime-t_1^\prime>\tau$, which contradicts with Equation~\ref{eq:test}. Hence, $Occ_{OPT}^\prime(e,S)$ cannot satisfy the time constraint $\tau$.
\end{proof}

Suppose $Occ_{OPT}(e,S)$, which consists of $[t_1,\ldots,t_k]$ ($t_1<\ldots<t_k$), fails to satisfy time constraint $\tau$, according to the above theory, it is easy to know that any entry $t_i^\prime\le t_i$ in list $L_i$ cannot be used to generate an occurrence that passes time constraint test. Therefore, if the minimum occurrence $Occ_{OPT}(e,S)$ fails to satisfy time constraint $\tau$, we need to eliminate the entries $t_i^\prime\le t_i$ in each layered list $L_i$ for all $1\le i\le k$. Moreover, we need to eliminate some other entries in $L_i$ to guarantee that $\min L_1<\min L_2<\ldots<\min L_k$, as those entries not satisfying this property are also useless for minimum occurrence extraction. Besides, the active layer is updated accordingly after the elimination.

As shown in Figure~\ref{fig:runexm}(13) and Figure~\ref{fig:runexm}(13+), the occurrence of the episode (\ie circled red entries in Figure~\ref{fig:runexm}(13)), namely $\langle9,12,13\rangle$, fails to pass the time constraint test as $13-9>3$. Then, in each $L_i$, we eliminate all the entries $t_i^\prime\le t_i$. In $L_1$, we eliminate $9$ and any other entries before that. Similarly, in $L_2$ and $L_3$, we eliminate $12$ and $13$, respectively. The other entries, \ie $12$ in $L_1$, is left in $OM(e)$, the rest lists are all empty again. Therefore, we reset $OM(e).\ell=2$.

With the help of all these operations above, an OccMap has the following interesting features.
\begin{itemize}
  \item An OccMap corresponds to exactly one time-constrained serial episode.
  \item The last/bottom layered list contains no more than one entry; there is at most one occurrence for the target time-constrained serial episode.
  \item The entries in all the layered lists satisfy: $\min L_1 < \ldots < \min L_k$. On the other hand, as timestamps appended into the same list strictly follow time sequence, $\min L_i=L_i[1]$. Therefore, the above property can be rewritten as $L_1[1]<\ldots<L_k[1]$.
  \item Notably, to save memory, after each \emph{list update} operation, we additionally check the inserted timestamp, say $S[j].t$, against the first entry in $L_i$, if $S[j].t-L_i[1]>\tau$, $L_i[1]$ cannot be used to generate any minimum occurrence that satisfies time constraint $\tau$. Therefore, in this case $L_i[1]$ is also be eliminated during \emph{list update}. In this way, the size of each list $L_i$ is in fact upper bounded by $\tau\rho$ if the event in the stream arrives with a constant speed $\rho$.
\end{itemize}

\renewcommand{\algorithmiccomment}[1]{\hskip3em$/*$ #1 $*/$}
\begin{algorithm}[tb]
\caption{ONCE algorithm}
\label{alg:Framwork}
\begin{algorithmic}[1]
\REQUIRE Streaming sequence $S=\langle (s_1,t_1), (s_2,t_2), \ldots, (s_n,t_n),\ldots\rangle$, target time-constrained serial episode $e^\tau=\langle \phi_1,\phi_2,\dots,\phi_k\rangle$
\ENSURE $freq(e^\tau,S)$: the non-overlapped frequency of $e^\tau$ in $S$
\STATE Initialize $OM(e^\tau)$ for $e$ with $k$ empty lists $L_1,\ldots,L_k$ that are hierarchically organized
\FOR {each event $S[i]$ in $S$ arrives}
    \STATE {$ListUpdate(OM(e^\tau),S[i])$}
    \COMMENT {Algorithm~\ref{alg:update}}
    \IF {the bottom layer is not empty}
        \STATE {$flag\leftarrow Validate\&Eliminate(OM(e^\tau))$}
        \COMMENT {Algorithm~\ref{alg:valeli}}
        \IF {$flag$ is TRUE}
            \STATE {$freq(e^\tau,S)++$}
        \ENDIF
    \ENDIF
\ENDFOR
\RETURN $freq(e^\tau,S)$
\end{algorithmic}
\end{algorithm}
\vspace{0ex}\subsection{The complete ONCE algorithm}\label{ssec:once}
Figure~\ref{fig:runexm} shows the complete one-pass process of mining the non-overlapped frequency for time-constrained serial episode $e^3=\langle A,A,B\rangle$ within $S$ shown in Example~\ref{ex2}. Notably, all the other events that do not appear in $e^3$ (\ie $C$) are ignored in the process. Each successful time constraint test over the minimum occurrences is marked in green, the unsuccessful ones are marked in red. It is easy to know from the figure that the non-overlapped frequency for $e^3$ in $S$ is $3$, \ie the number of successful time constraint tests. The detailed algorithms are shown in Algorithm~\ref{alg:Framwork},~\ref{alg:update} and~\ref{alg:valeli}.

In Algorithm~\ref{alg:Framwork}, given the input of streaming sequence $S$ and target time-constrained serial episode $e^\tau$, ONCE algorithm first initializes an OccMap for $e^\tau$ (Line 1). As each event in the stream arrives, ONCE first performs \emph{list update} (\ie Algorithm~\ref{alg:update}) based on the event (Lines 2-3). When the last layer in OccMap is not empty, we perform \emph{occurrence validation} to find the minimum occurrence and test whether it satisfies $\tau$. \emph{Invalid entries elimination} is performed immediately after that (Lines 4-5 and Algorithm~\ref{alg:valeli}). If the time constraint test successes, the frequency of $e^\tau$ is increased by $1$ (Lines 6-7). Finally, the frequency is returned (Line 11).

Algorithm~\ref{alg:update} works as follows. Given the OccMap $OM(e^\tau)$ to be updated and event $S[i]$, we check every activated list $L_j$ which wait for update, if $S[i].e$ matches the event corresponding to the $L_j$, we append $S[i].t$ to the end of $L_j$ (Lines 2-3). Afterwards, we perform a local check in $L_j$ in order to eliminate out-of-date entries (\ie old entries that cannot constitute a minimum occurrence that satisfies $\tau$) from $L_j$ (Lines 4-6). Finally, we update the active layer to the next empty list (Line 9). Obviously, the time complexity for Algorithm~\ref{alg:update} is $O(k)$ where $k$ is the length of $e^\tau$.

In Algorithm~\ref{alg:valeli}, we first extract the minimum occurrence from $OM(e^\tau)$ (Lines 1-5) and then eliminate invalid entries (Lines 6-19). To extract the minimum occurrence, we first set the right bound of the occurrence interval $t_k$ as $L_k[1]$, which is the only entry in $L_k$ (Line 1). Afterwards, we iteratively find from each upper layer $t_i$ as the latest timestamp that appears before $t_{i+1}$. This process continues until $t_1$ (Lines 2-5). Therefore, $[t_1,\ldots,t_k]$ constitute a minimum occurrence for $e^\tau$. The time complexity of this process is $O(k\log\overline{|L|})$\footnote{\scriptsize In the implementation, as each $t_i$ always locates in the end of $L_i$, we alternatively use linear search from the end of $L_i$ in Line 3, the average time of which is better than binary search in practice. Similar strategy also applies to Line 15.}. As the length of each list in $OM(e^\tau)$ is in fact upper bounded by $\tau\rho$ if the event in the stream arrives with a constant speed $\rho$. According to Algorithm~\ref{alg:update}, in Line 4 and 5, we know that $(L_j[k].t-L_j[1].t)<\tau$. Considering the fact that $\rho$ is a positive number, as a result, $\rho(L_j[k].t-L_j[1].t)<\tau\rho$. Actually, considering the ListUpdate process, we have $k\leq \rho(L_j[k].t-L_j[1].t)$. With the two inequalities above, we can conclude that $k<\tau\rho$, which proves that the upper bound of $L_j$ in $OM(e^\tau)$ is $\tau\rho$. The time complexity is in fact $O(k\log\tau)$ if $\rho$ is fixed. Afterwards, we test whether the extracted minimum occurrence satisfies the time constraint. If the test successes, we eliminate all entries from the OccMap and return (Lines 6-9). Otherwise, we eliminate all entries in each list $L_i$ subject to $L_i[j]\le t_i$ (Lines 11-13). Besides, we also need to make sure $L_1[1]<\ldots<L_{k-1}[1]$ by eliminating some other entries, as those entries are also useless in constituting further minimum occurrences (Lines 14-16). Finally, we update the active layer as the first empty layered list (Lines 17-18). It is easy to see the complexity of elimination is also $O(k\log\tau)$ as searching from a list an entry (Line 3 and 15) takes $O(\log\tau)$ using binary search. Therefore, the complexity of Algorithm~\ref{alg:valeli} is $O(k\log\tau)$.

In all, the time complexity for ONCE algorithm in Algorithm~\ref{alg:Framwork} to process a single event is $O(k\log\tau)$. Taking into account the number of events in the stream $S$, the time complexity of processing $n$ events is all together $O(nk\log\tau)$.

\begin{algorithm}[tb]
\caption{ListUpdate algorithm}
\label{alg:update}
\begin{algorithmic}[1]
\REQUIRE OccMap $OM(e^\tau)$ and the next event in stream sequence $S[i]$
\ENSURE updated OccMap $OM(e^\tau)$
\FOR {all activated lists $L_j$ ($j\le OM(e^\tau).\ell$) and $S[i].e\in e^\tau$}
    \IF {$S[i].e=L_j.e$}
        \STATE {append $S[i].t$ to the end of $L_j$}
        \IF {$S[i].t-L_j[1]>\tau$}
        \STATE {remove $L_j[1]$ from $L_j$}
        \ENDIF
    \ENDIF
\ENDFOR
\STATE Update $OM(e^\tau).\ell$ to the next empty list
\RETURN $OM(e^\tau)$
\end{algorithmic}
\end{algorithm}

\begin{algorithm}[tb]
\caption{Validate\&Eliminate algorithm}
\label{alg:valeli}
\begin{algorithmic}[1]
\REQUIRE OccMap $OM(e^\tau)$
\ENSURE $flag$ indicating whether the occurrence passes time constraint test
\STATE $t_k\leftarrow L_k[1]$
\FOR{$i$ in $k-1$ to $1$}
    \STATE {find the maximum $j$ subject to $L_i[j]<t_{i+1}$}
    \STATE {$t_i\leftarrow L_i[j]$}
\ENDFOR
\IF {$t_k-t_1\le\tau$}
    \STATE {remove all entries from $OM(e^\tau)$}
    \STATE {reset $OM(e^\tau).\ell$ as $1$}
    \RETURN TRUE
\ELSE
    \FOR{all lists $L_i$ in $OM(e^\tau)$}
        \STATE {remove from $L_i$ all entries $L_i[j]\le t_i$}
    \ENDFOR
    \FOR{$i$ in $2$ to $k-1$}
        \STATE {remove from $L_i$ all entries $L_i[j]\le L_{i-1}[1]$}
    \ENDFOR
    \STATE {reset $OM(e^\tau).\ell$ to $\max\limits_{|L_i|>0}{i}+1$}
    \RETURN FALSE
\ENDIF
\end{algorithmic}
\end{algorithm}

\vspace{0ex}\subsection{Correctness of ONCE algorithm}\label{ssec:correct}
In this part, we discuss the correctness of ONCE algorithm. In particular, we need to show that ONCE algorithm can correctly answer the problem in Definition~\ref{def:pfemin}, namely counting the non-overlapped frequency $freq(e^\tau,S)$ of given time-constrained serial episode $e^\tau$ as the event in streaming sequence $S$ passes by. To this end, we present a pair of lemmas below, based on which we can finally prove the correctness of ONCE.

\vspace{0ex}\begin{lemma}\label{lemma1} Suppose the frequency of $e^\tau$ in $S$ returned by ONCE is denoted by $freq'(e^\tau,S)$ and the ground truth is denoted by $freq(e^\tau,S)$, then $freq'(e^\tau,S)\le freq(e^\tau,S)$.
\end{lemma}
\begin{proof}
Given in Appendix~\ref{app1}.
\end{proof}

\vspace{0ex}
\begin{lemma}\label{lemma2} Suppose the frequency of $e^\tau$ in $S$ returned by ONCE is denoted by $freq'(e^\tau,S)$ and the ground truth is denoted by $freq(e^\tau,S)$, then $freq'(e^\tau,S)\ge freq(e^\tau,S)$.
\end{lemma}
\begin{proof}
Given in Appendix~\ref{app2}.
\end{proof}

\vspace{0ex}\begin{theorem}\label{thm:correct} ONCE algorithm (Algorithm~\ref{alg:Framwork},~\ref{alg:update} and~\ref{alg:valeli}) can correctly answer the time-constrained frequency counting problem in Definition~\ref{def:pfemin}.
\end{theorem}
\begin{proof} Directly follows Lemma~\ref{lemma1} and~\ref{lemma2}.
\end{proof}

\vspace{0ex}\section{ONCE+ ALGORITHM}\label{sec:once+}
In previous section we have presented ONCE algorithm, ONCE can compute the overlapped frequency of target time-constrained serial episode. In fact, it can also be adapted to compute distinct frequency with series of modifications. In this part, we propose the modified version, namely ONCE+, to compute the distinct frequency of time-constraint serial episodes.

ONCE+ also utilizes OccMap to store the timestamps of events in target serial episode, the only difference is in \emph{Validate\&Eliminate} algorithm. In order to distinguish ONCE and ONCE+, we name the modified algorithm as \emph{Validate\&Eliminate+} algorithm. The detailed algorithm is shown in Algorithm~\ref{alg:valeli+}.

In Algorithm~\ref{alg:valeli+}, we need to find the occurrence which meets the time-constrained condition from $OM(e^\tau)$. Firstly, we compare all the entries of $L_1$ with $L_k[1]$. If we failed find the minimum entry $L_1[i]$ which meets $L_k[1]-L_1[i]<\tau$, eliminate all entries from the OccMap and return (Lines 26-29). If we find it, set the first layer $t_1$ as $L_1[i]$. Afterwards, we iteratively find from each upper layer $t_i$ as the latest timestamp that appears before $t_{i+1}$. This process continues until $t_k$ (Line 1-9). If we failed to find $t_i$ ($1<i<k$) in any layer, we also have to eliminate all entries from the OccMap and return (Line 21-25). Therefore, $[t_1,\ldots,t_k]$ constitute a distinct occurrence for $e^\tau$. Then, we eliminate all the entries from OccMap which is no later than $t_i$ (Line 10-20). Finally, we update the active layer as the first empty layered list.

Notably, in ONCE we find $t_i$ from bottom to top of the OccMap, then test whether the extracted minimal occurrence satisfies the time constraint condition. However, in ONCE+, we firstly need to find the occurrence from $OM(e^\tau)$ which meets the time-constrained. Then find $t_i$ from top to bottom of the OccMap. It is easy to see, the only difference between ONCE and ONCE+ is that the process of seeking $t_i$ from OccMap is opposite, so the complexity of Algorithm~\ref{alg:valeli+} is the same as Algorithm~\ref{alg:valeli}.

\begin{algorithm}[tb]
\caption{Validate\&Eliminate+ algorithm}
\label{alg:valeli+}
\begin{algorithmic}[1]
\REQUIRE OccMap $OM(e^\tau)$
\ENSURE $flag$ indicating whether the occurrence passes time constraint test

\IF {find the minimum i subject to $L_k[1]-L_1[i]<\tau$}
    \STATE $t_1\leftarrow L_1[i]$
    \STATE {p=0}
    \FOR{$i$ in $2$ to $k$}
    	\IF {find the minimum j subject to $L_i[j]>t_{i-1}$}
    	\STATE {$t_i\leftarrow L_i[j]$}
    	\STATE {P++}
    	\ENDIF
    \ENDFOR
    \IF {p=k-1}
    	\FOR{all list $L_i$ in $OM(e^\tau)$}
    		\STATE {remove from $L_i$ all entries $L_i[j]\le t_i$}
    		\FOR {$y$ in $1$ to $k$}
    			\IF {find y subject to $L_i[j]=t_y$}
    				\STATE {remove $L_i[j]$ from $L_i$}
    			\ENDIF
    		\ENDFOR
    	\ENDFOR
    	\STATE {reset $OM(e^\tau).\ell$ to $\max\limits_{|L_i|>0}{i}+1$}
    	\RETURN TRUE
    \ELSE
    	\STATE {remove all entries from $OM(e^\tau)$}
    	\STATE {reset $OM(e^\tau).\ell$ as $1$}
    	\RETURN FALSE
    \ENDIF
\ELSE
    \STATE {remove all entries from $OM(e^\tau)$}
    \STATE {reset $OM(e^\tau).\ell$ as $1$}
    \RETURN FALSE
\ENDIF
\end{algorithmic}
\end{algorithm}

In order to illustrate the characteristics of ONCE+ algorithm. We use another example to show how ONCE+ works. In Figure~\ref{fig:once+fig}, there is an OccMap that corresponds to $e^\tau=\langle A,A,B\rangle$.
\begin{equation*}
\begin{split}
S=\langle(A,1),(C,2),(A,3),(D,4),(A,5),(C,6),(A,7),\\
(C,8),(B,9),(B,10)\rangle
\end{split}
\end{equation*}
when $\tau=9$, we follow Algorithm~\ref{alg:valeli+}, $\langle 1,3,9\rangle$ and $\langle 5,7,10\rangle$ are distinct occurrences of $e^9=\langle A,A,B\rangle$, so $freq^+(e^9,S)=2$, if we set $\tau=7$, the distinct occurrences of $e^7=\langle A,A,B\rangle$ are then show (in green) in Figure~\ref{fig:once+fig2}, that is $freq^+(e^7,S)=1$.

However, if we use ONCE to count the non-overlapped frequency of $e^9=\langle A,A,B\rangle$ or $e^7=\langle A,A,B\rangle$, the only non-overlapped minimal occurrence is $\langle 5,7,9\rangle$, that is, $freq(e^7,S)=1$ and $freq(e^9,S)=1$.

The correctness of ONCE+ is easy to justify following the same way as Section~\ref{ssec:correct}.
\begin{figure}
\centering
  \epsfig{file=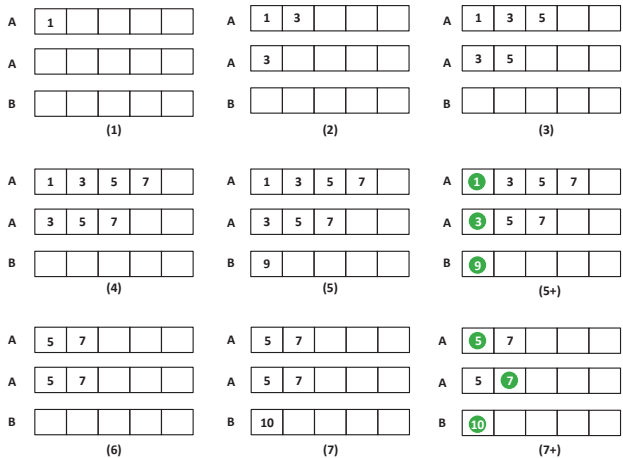, width=\columnwidth}

\vspace{0ex}\caption{Mining the non-overlapped frequency with intersected occurrence of $e^9=\langle A,A,B\rangle$ in the sequence S [best viewed in color].}
\vspace{0ex}\label{fig:once+fig}       
\end{figure}

\begin{figure}
\centering
  \epsfig{file=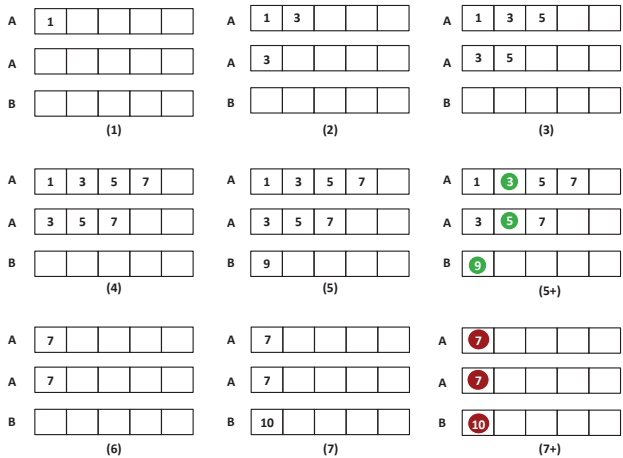, width=\columnwidth}

\vspace{0ex}\caption{Mining the non-overlapped frequency with intersected occurrence of $e^7=\langle A,A,B\rangle$ in the sequence S [best viewed in color].}
\vspace{0ex}\label{fig:once+fig2}       
\end{figure}

\vspace{0ex}\section{Experimental Study}\label{sec:exp}
In this part, we conduct experimental study over both synthetic and real world data. Through the experimental results, we justify that ONCE (\resp ONCE+) can answer the non-overlapped (\resp distinct) frequency counting problem for time-constrained serial episodes in a one-pass way efficiently. The real world data is the streaming sequence of telecommunication alarms within 4 cities in Guizhou Province of China during 2014. Besides, we also generate a synthetic dataset by randomly sample an event at each timestamp. The statistics of both datasets are shown in Table~\ref{tab:dataset}. The synthetic data are generated by uniformly randomly sampling a particular event from $\Sigma$ one step after another. All the experiments are tested on a workstation with Xeon E5-2603v3 1.6GHz CPU, 16GB RAM running Ubuntu 12.04 LTS. We compare ONCE algorithm and ONCE+ with other baselines, namely SASE+ and SASE++~\cite{DBLP:conf/sigmod/ZhangDI14}. All the parameters in SASE+ and SASE++ are optimized according to their suggested settings.

Notably, in the following experiments, we report the average throughput, which is defined as the number of signals processed by an algorithm per second. In line with~\cite{DBLP:conf/sigmod/ZhangDI14}, we report how the throughput can be affected by different factors. Through all these experiments, we are exciting to find that ONCE only takes less than $1\mu s$ for each $S[i]$, especially for the real world dataset. That is, our model can work in event-intensive stream even if millions of events arrive in single second\footnote{\scriptsize The code and dataset will be released once this work is published.}.
\begin{table}
\vspace{0ex}\caption{Dataset statistics}
\label{tab:dataset}       
\centering
\vspace{0ex}
\begin{tabular}{llll}
\hline\noalign{\smallskip}
datasets & $|S|$ & $\Sigma$ & duration  \\
\noalign{\smallskip}\hline\noalign{\smallskip}
Telecom. alarms & 8,821,220 & 252 & 2014-05-01 0h to 2014-05-31 24h \\
Synthetic & 91,021 & 622 & - \\
\noalign{\smallskip}\hline
\end{tabular}
\vspace{0ex}\end{table}

\vspace{0ex}\btitle{Selectivity $\theta$.} It is defined as, $\frac{\#Matches}{\#Events}$, which is controlled by changing the target episodes in stream $S$\footnote{Once a timestamp is inserted into OccMap, we identify it as a \emph{Match}.}. Similar to~\cite{DBLP:conf/sigmod/ZhangDI14}, it is varied from $10^{-6}$, up to 1.6, which is a very heavy workload to test our algorithms. We simulate the stream $S$ by sequentially input a new signal after some time interval. In particular, for the real-world data, each signal in the experiment arrives exactly the same with its original time interval; for synthetic dataset, we set each signal arrives with a constant speed every 1$ms$. Firstly we test the \emph{throughput} processing all signals in $S$, and report the average over all 10 episodes. Figure~\ref{fig:exp2} show the throughput of the real-world data and synthetic dataset while varying $\theta$. We see that the throughput of SASE+ drops very fast as $\theta$ increase, and that of SASE++ is worse than ONCE and ONCE+. The throughput of ONCE and ONCE+ is similar. SASE++, ONCE and ONCE+ are not sensitive to the selectivity. The throughput of ONCE and ONCE+ is nearly an order of magnitude better than SASE++.

\vspace{0ex}\btitle{Effect of $\tau$.} Secondly, we test the throughput of ONCE and ONCE+ by varying the time constraint for the given serial episodes. In particular, at each $\tau$ we randomly select 10 different episodes with length $5$. We test the throughput of Algorithm~\ref{alg:Framwork} processing all signals in $S$, and report the average over all 10 episodes. Notably, as in synthetic data each signal is associated with a discrete step, $\tau$ is defined as the maximal number of steps an episode should cover. In real-world dataset, we vary the time constraint from $0.5$ to $12$ hours. The results are shown in Figure~\ref{fig:exp1}. Notably, the response time for all the cases remains almost constant. The phenomenon seems different from our time complexity study. The reason is as follows. As $\tau$ is increased, the probability of performing Lines 7-9 in Algorithm~\ref{alg:valeli}, whose complexity is $O(1)$, will be mach larger than that of Lines 11-18, whose complexity is $O(k\log\tau)$. Therefore, as $\tau$ increases, the curve will tend to be more constant (as Lines 7-9 contributes more to the response time) than sublinear.

\begin{figure}[t]
\centering
\subfloat[][Synthetic]{\epsfxsize=0.48\columnwidth \epsffile{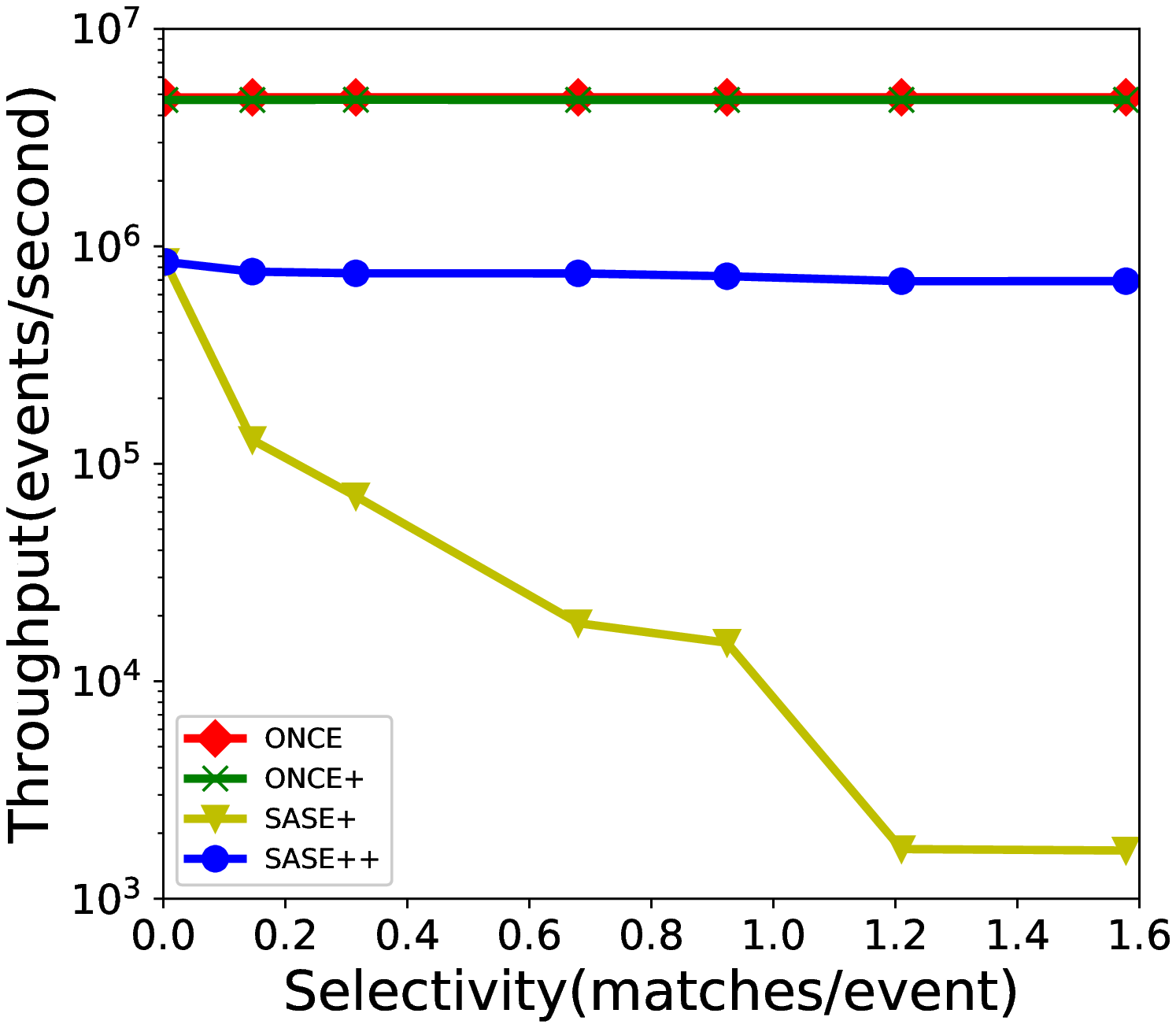}}
\subfloat[][Telecom. alarms]{\epsfxsize=0.48\columnwidth \epsffile{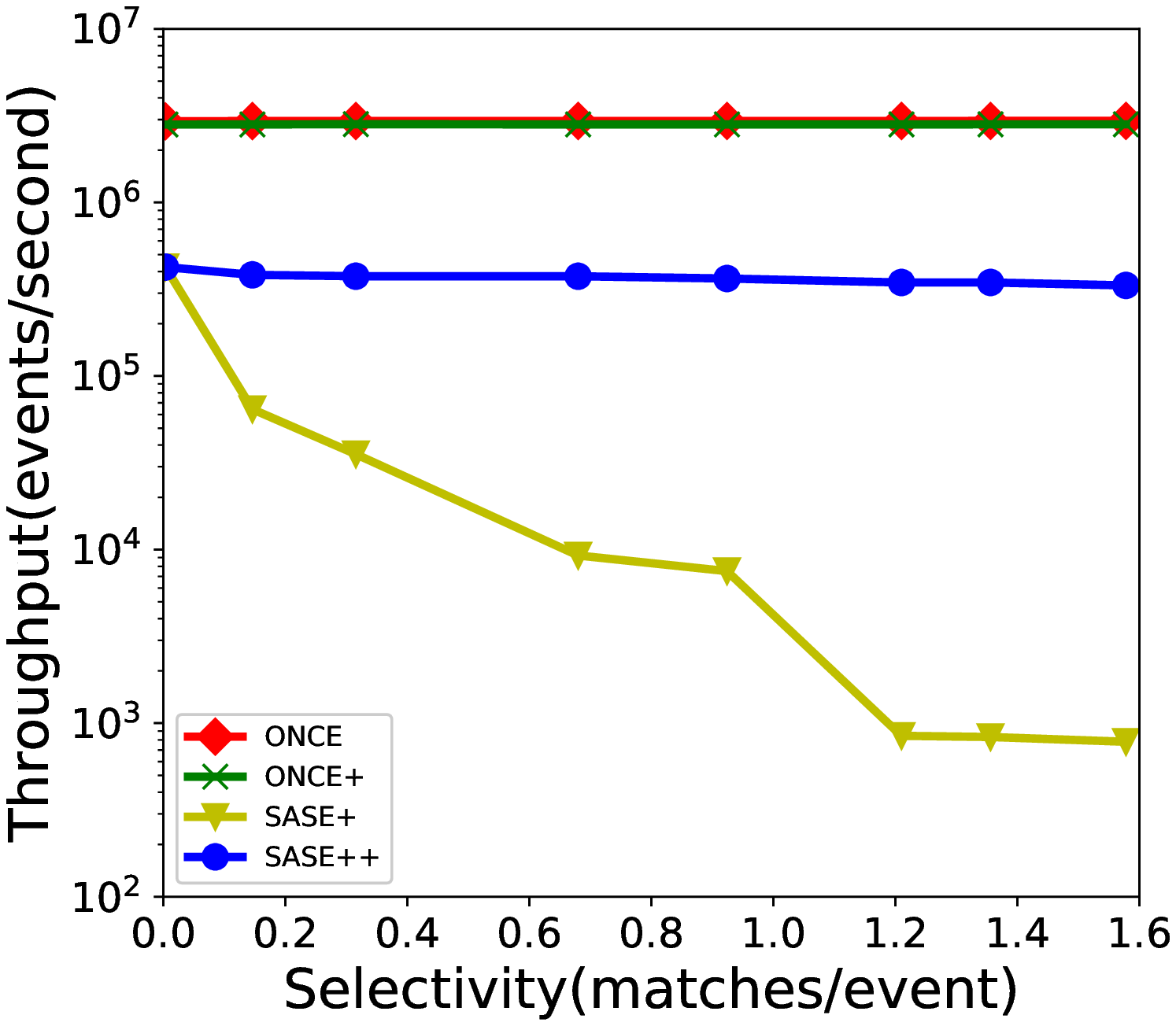}}
\vspace{0ex}\caption{The throughput by varying Selectivity (matches/event).}
\vspace{0ex}\label{fig:exp2}
\end{figure}

\begin{figure}[t]
\centering
\subfloat[][Synthetic]{\epsfxsize=0.48\columnwidth \epsffile{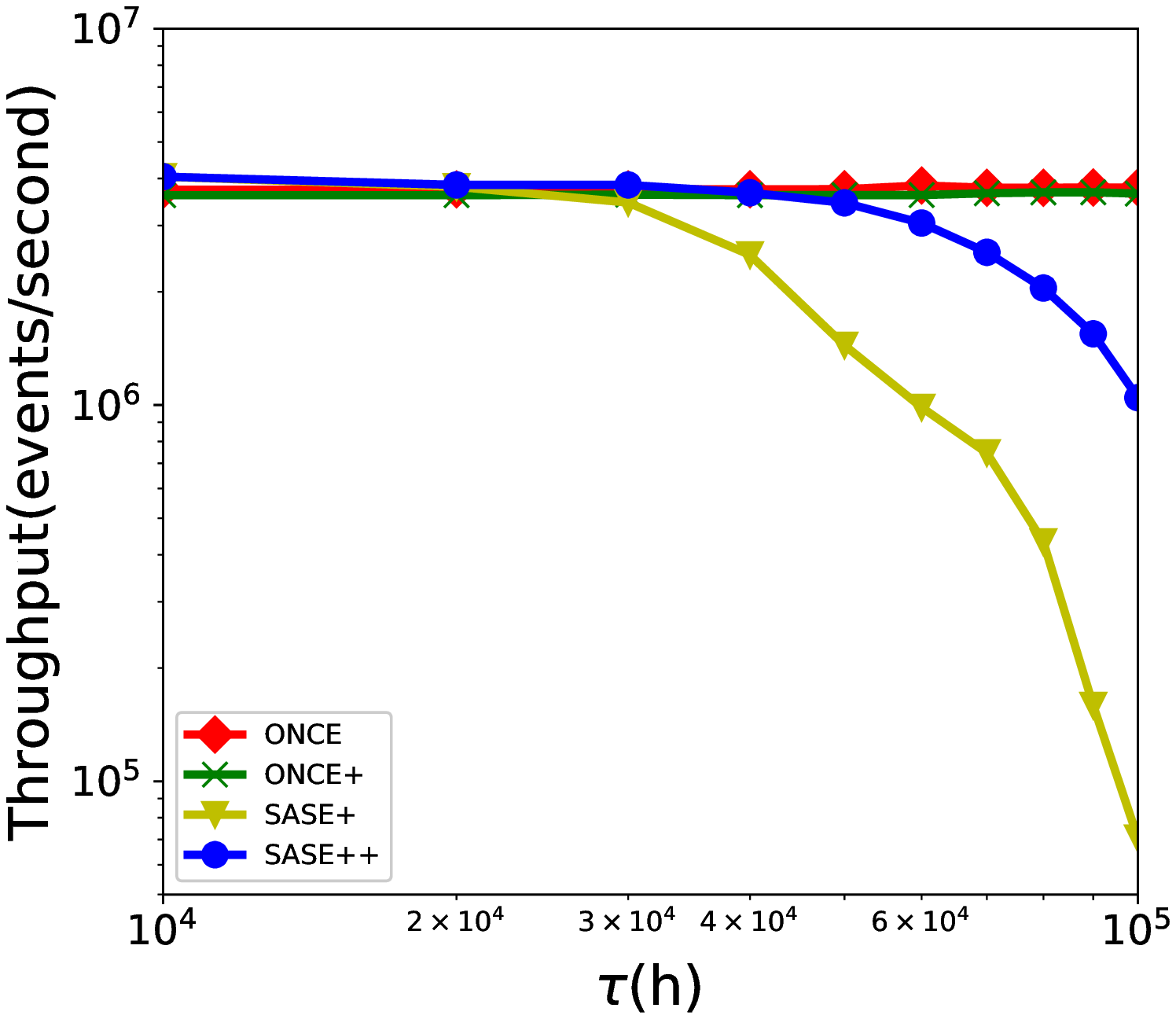}}
\subfloat[][Telecom. alarms]{\epsfxsize=0.48\columnwidth \epsffile{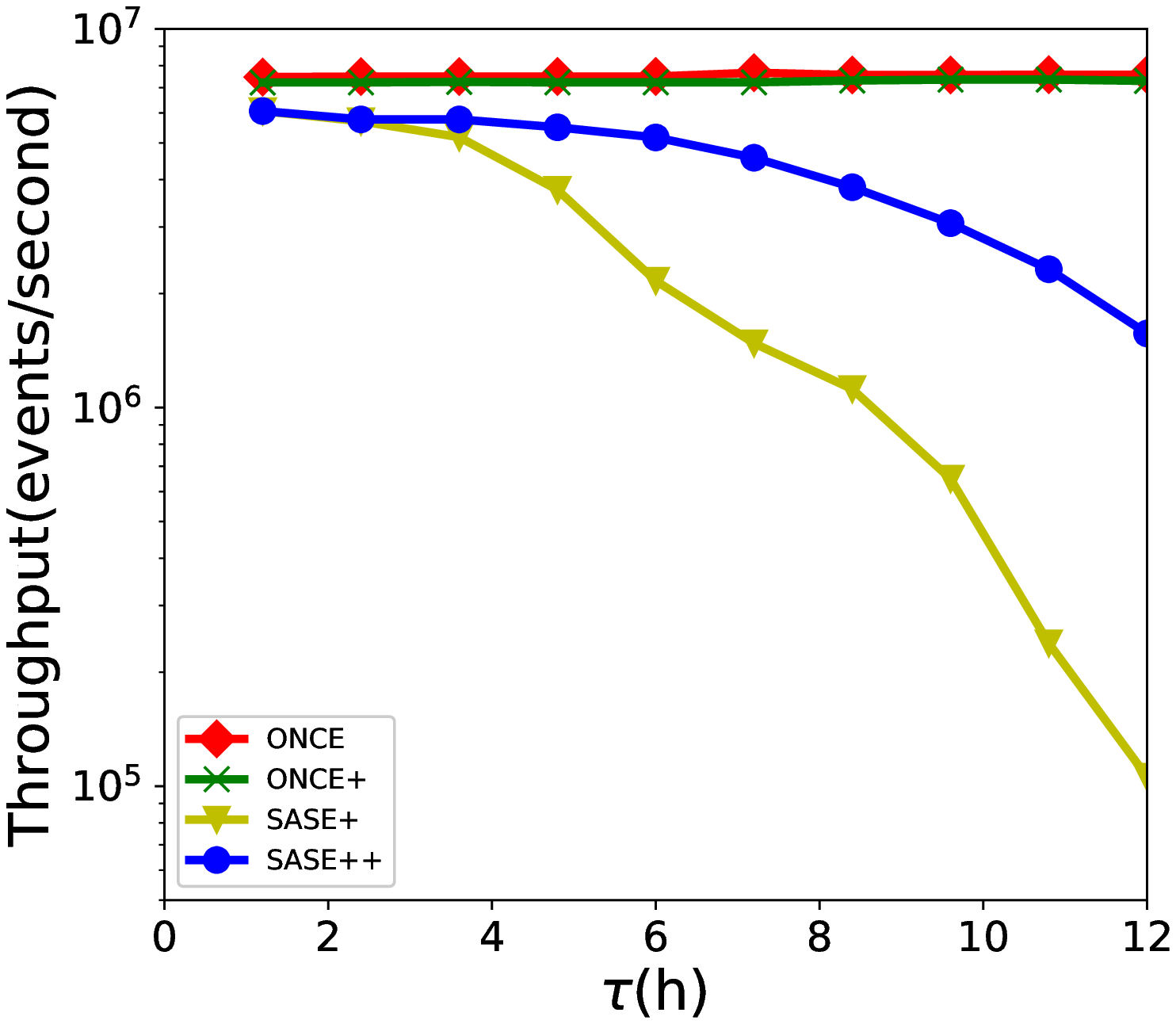}}
\vspace{0ex}\caption{The throughput by varying $\tau$.}\label{fig:exp1}
\vspace{0ex}\end{figure}

\vspace{0ex}\btitle{Implicit factors.} During the experimental study, we find that besides $\tau$ and $k$, the response time also vary for different $e^\tau$ even if they share the same length $k$ and $\tau$. The reason is that, according to Algorithm~\ref{alg:Framwork}, each time a new event $S[i]$ arrives, Lines 4-9 in Algorithm~\ref{alg:Framwork}, which is the most time-consuming, may not always be performed. Intuitively, each time $freq(e^\tau,S)$ is updated, this part is performed. Therefore, the frequency of $e^\tau$ implicitly affects the eventual response time. Hence, we conduct another experiment to test the effect of frequency by fixing both $k$ and $\tau$ at particular levels. The results are shown in Figure~\ref{fig:exp0}. We randomly select 10 episodes at each frequency level (\ie $500, 1000, 1500, 2000$) and report the average time for processing $S[i]$. We repeat the same setting for episodes with lengths $3,5,7$, respectively. As the maximum frequency for episodes with length $7$ is less than 1500, thus it does not appear when frequency is 1500 and 2000. Notably, the frequencies of episodes selected at each level (\eg $500$) may vary a bit (\eg $495, 502$ and etc.). Figure~\ref{fig:exp0} only reports that of the synthetic data, as we cannot find enough episodes at each frequency level in real world one. Obviously, the response time increases along with the frequency level, which agrees with our analysis above.

\begin{figure}[t]
\centering
\subfloat[][]{\epsfxsize=0.48\columnwidth \epsffile{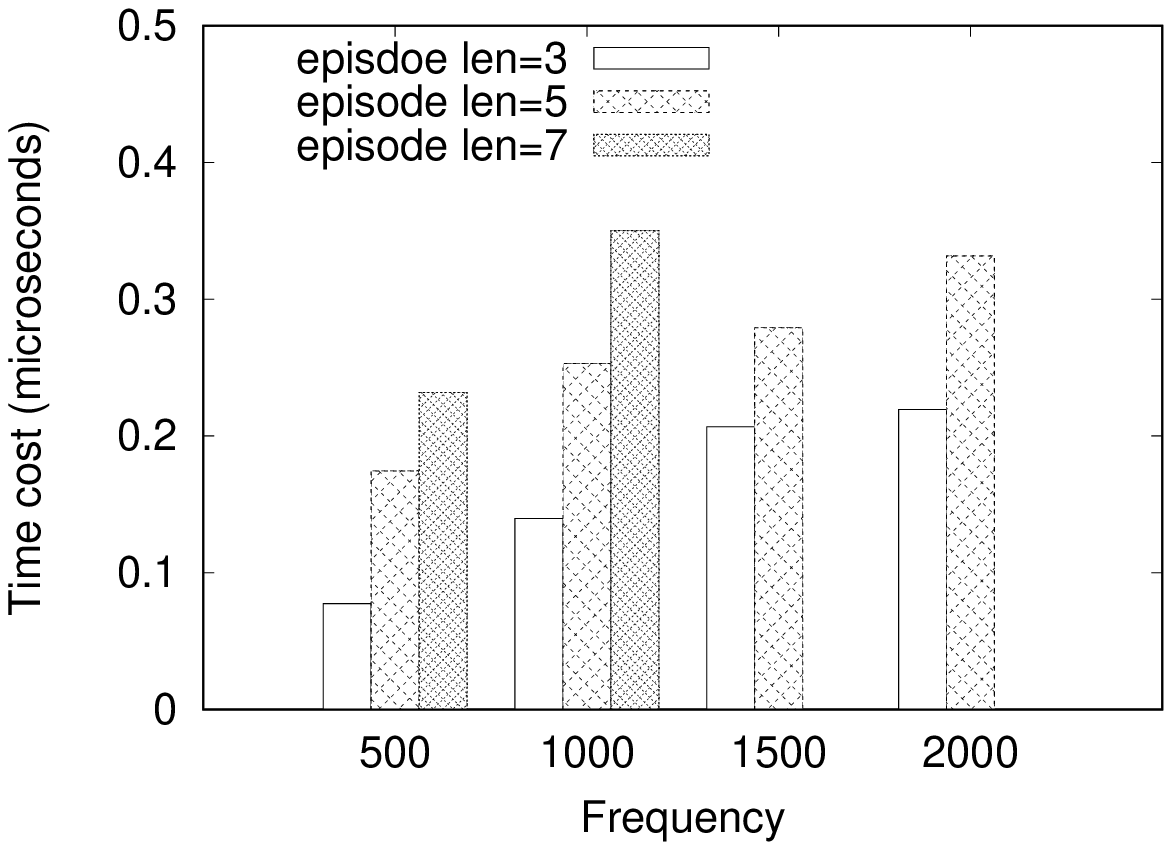}}
\subfloat[][]{\epsfxsize=0.48\columnwidth \epsffile{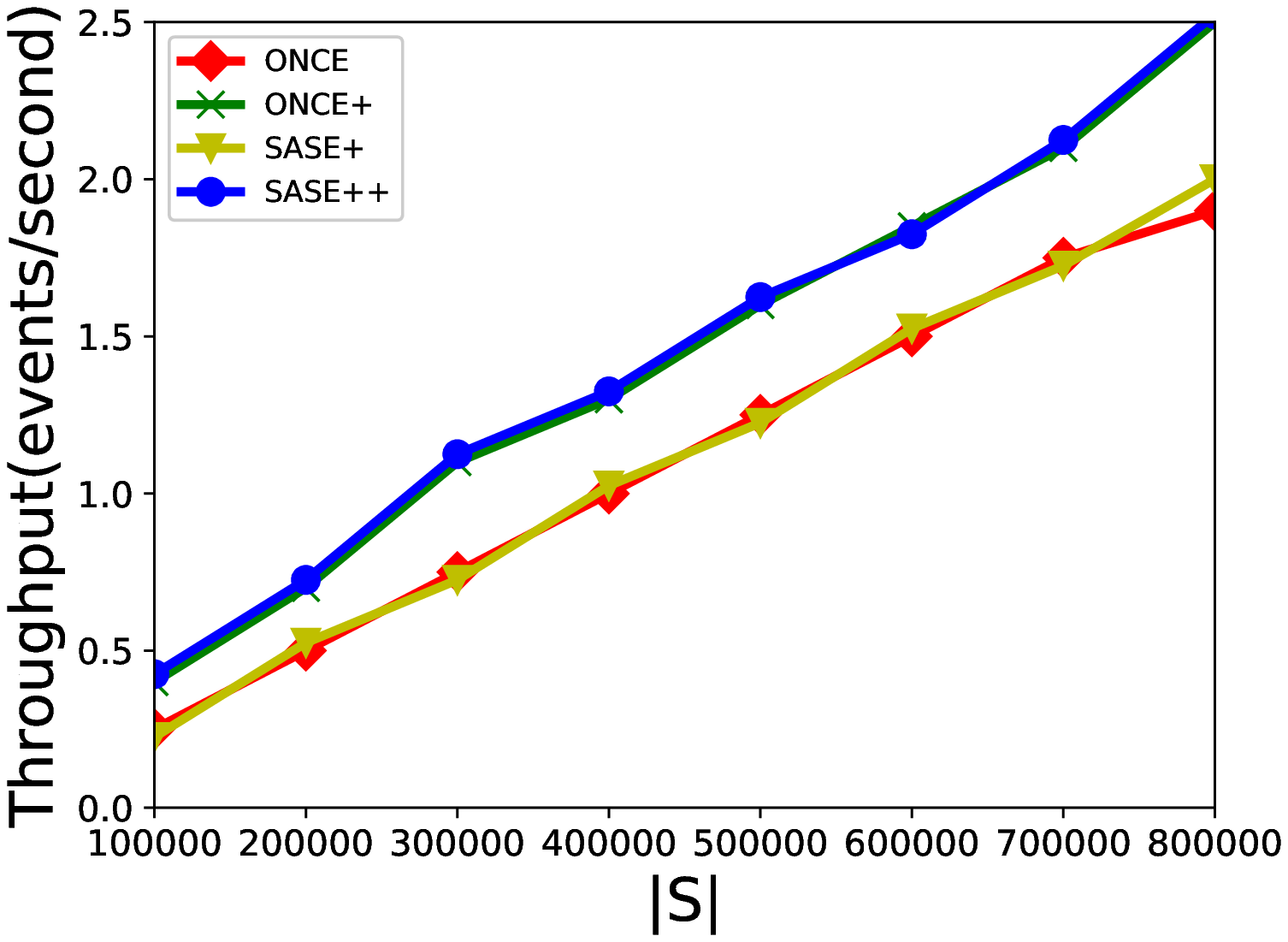}}
\vspace{0ex}\caption{(a) Effect of frequency; (b) Scalability.}\label{fig:exp0}
\vspace{0ex}\end{figure}

\eat{
\btitle{Comparison with baseline.} Finally, we compare ONCE, ONCE+ and one state-of-the-art time-constrained serial episode mining methods, namely DFS~\cite{DBLP:journals/dke/AcharAS13}. Methods such as~\cite{DBLP:conf/kdd/LaxmanSU07} that works in stream has already been justified inaccurate to address time-constrained problem in Section~\ref{ssec:prodef}, thus we do not take it into comparison. Notably,~\cite{DBLP:journals/dke/AcharAS13} is a multi-pass algorithm and cannot work in stream. For fairness, instead of reporting the response time for processing a single event $S[i]$ in a streaming, we report the total running time for reading all signals in a static sequence $S$ with $\tau$=3 hours. The results on running time are shown in Figure~\ref{fig:exp0}(b). Obviously, DFS takes more than 20 times the cost comparing to ONCE. Moreover, the running time of both DFS and ONCE depends on the length of $e^\tau$. Therefore, even if work in static mode, ONCE can also significantly outperform the state-of-the-art multi-pass methods. Besides, we also compare the result accuracy in the experiment. The output frequency for the given set of time-constrained episodes during the experiments are identical between both methods. That is, both DFS and ONCE can correctly output the frequency of time-constrained serial episode.
}

\vspace{0ex}\btitle{Scalability and memory consumption.} We now test the scalability of the model by varying the length of input sequence $S$ from $100,000$ to $700,000$. The response time are shown in Figure~\ref{fig:exp0}(b). It increases almost linearly, which is consistent with our analysis in the end of Section~\ref{ssec:once}. Notably, the average response time for processing a single signal is less than $1\mu s$. That is, ONCE and ONCE+ can work on signal-intensive streams where millions of events happen in a second. We also demonstrate how the memory usage of the core structure $OM(e^\tau)$ in ONCE algorithm scales with the size of target episodes. As described in $ListUpdate$ operation, each $OM(e^\tau)$ is initialized as $k$ layered empty lists. As massive data flow in, corresponding signal will be stored in this structure. To evaluate the memory consumption of the proposed algorithm during this dynamic process, the maximum cost (the structure reaches its largest condition at occurrence validation step) is measured under the condition that $k$ is steadily increased from 3 to 11. Notably, at each particular length level (3, 5, 7, 9, 11), we randomly select 10 target episodes whose occurrences are count and corresponding memory consumptions are evaluated. As is evident in Table~\ref{tab:mem}, the memory consumption grows with $k$ because the number of lists in $OM(e^\tau)$ increases with the length of the target episode. Therefore, with more layered lists to generate candidate occurrence, $OM(e^\tau)$ shall store more signals before occurrence validation. Notably, SASE++ exhibits the same memory consumption with SASE+, and the memory consumption of ONCE+ is the same to ONCE, so we only list the comparative results between ONCE and SASE+.

\begin{table}[htbp]
  \centering

  \caption{Memory consumption by varying $k$ (in KB)}\label{tab:mem}

    \begin{tabular}{r|r|r|r|r}
    \toprule
          & \multicolumn{2}{c|}{Synthetic} & \multicolumn{2}{c}{Telecom. Alarms} \\
\cmidrule{2-5}    \textit{k} & \multicolumn{1}{c|}{ONCE} & \multicolumn{1}{c|}{SASE+} & \multicolumn{1}{c|}{ONCE} & \multicolumn{1}{c}{SASE+} \\
    \midrule
    3     & 0.1089 & 1.1194 & 0.1338 & 1.2403 \\
    5     & 0.4393 & 2.282 & 0.8174 & 4.3964 \\
    7     & 0.7067 & 4.2985 & 1.1613 & 10.1032 \\
    9     & 1.7303 & 14.3811 & 7.4654 & 66.4988 \\
    11    & 2.0011 & 20.198 & 14.3659 & 162.8722 \\
    \bottomrule
    \end{tabular}%
  \label{tab:addlabel}%
\end{table}%

\vspace{0ex}\section{Application and Observations}\label{sec:obs}
As the direct implementation of the algorithm, it has been applied in practice within a corresponding telecom alarm management system within China. The system is built to monitor all the alarms sent from network equipments of a particular service provider within Guizhou province in real time. It is expected that through the system, operators and officers can respond to system faults and emergency as soon as possible.

For instance, the following is a standard rule that is listed in the Service Manual of the system operator, ``If more than 10\% of the cells in a district are out-of-service, local officer have to be dispatched to handle it''. There exist hundreds of other similar rules within the manual. Applying these rules in practice requires counting the specific incidents in real-time while the alarms keep on arriving in streaming manner. Unfortunately, this is not a trivial task, as the majority of the incidents appeared in the listed rules cannot be simply interpreted as a single alarm. Instead, they can only be interpreted as a sequential combination of series of particular alarms. For example, according to the empirical study of the company, a cell is ``out-of-service'' only if there are three alarms, namely `\emph{low voltage}', `\emph{base station disconnect}' and `\emph{carrier wave alarm}', appearing sequentially within 3 minutes. Therefore, in order to respond the the incidents effectively and efficiently, the system has to be able to count the frequency of specific serial combination of alarms in real-time and respond as soon as possible whenever the frequency is beyond a predefined threshold. Table~\ref{tab:rules} lists the incidents, their corresponding sequential alarms as well as the respond thresholds.

To address the problem, we are invited to apply ONCE and ONCE+ algorithms in this system over the streaming alarms received in real-time. Within this practical applications, there are over 250,000 alarms received everyday, that is, more than 3 alarms every second. In another word, the system has to observe an arbitrary incident listed in the manual within $1/3$ second, otherwise the system will fail to respond properly. ONCE and ONCE+ algorithm deployed in this system can successfully output the frequency of required incidents (\resp time-constrained serial alarms) within a millisecond, which has also been demonstrated in our experimental study. Notably, most of the serial alarms (shown in Table~\ref{tab:rules}) are restricted to happen in the same district (\resp NE, base station), which, in fact, makes each individual alarm signal a two-dimensional sample such that ONCE and ONCE+ cannot be directly applied. To address that, the same alarm happening in different places (\ie NE, base station) are treated as different symbols in $\Sigma$, that is, they are viewed as completely different symbols in ONCE and ONCE+. In this way, the frequency of all the required incidents can be counted in the system.

\begin{table*}[htbp]
  \centering
  \caption{Incidents and the interpreted serial alarms in the Service Manual (part).}
    \begin{tabular}{|l|p{7cm}|l|p{3cm}|}
    \toprule
    \textbf{Incidents} & \textbf{Corresponding serial alarms ($e$)} & \textbf{Time thresholds ($\tau$)} & \textbf{Frequency thresholds} \\
    \midrule
    cell out of service & low voltage, base station disconnect, carrier wave alarm & 3 minutes & 10\% of all cells in a district \\
    \hline
    AP out of service & ap reboot, ap fault & 5 minutes & 5\% of all cells in a district \\
    \hline
    APG process failure 1 & APG process reinitiated, statistics and traffic measurement colleciton timeout fault, CPT fault & 5 minutes & 5\% of all cells in a district \\
    \hline
    APG process failure 2 & APG process reinitiated, CPT fault, statistics and traffic measurement colleciton timeout fault & 5 minutes & 5\% of all cells in a district \\
    \hline
    NE out of service & SNT fault, MSC fault, MSC fault & 5 minutes & 10 in the same NE \\
    \hline
    base station out of service & carrier wave alarm, DTS fault, DS fault & 15 minutes & 5 in the same base station \\
    \hline
    NE communication blocked & Digital path unavailible state fault, Digital path fault supervision & 5 minutes & 5 in the same NE \\
    \hline
    NE synchronous failure & Synchronous digital path fault supervision, Synchronous fault & 5 minutes & 5 in the same NE \\
    \hline
    gateway articulate failure 1 & GARP fault, MGW fault, gateway channel fault, gateway block alarm & 5 minutes & 5\% of all gateways in the same district \\
    \hline
    gateway articulate failure 2 & SER reinitiated, MGW fault, gateway channel fault, gateway block alarm & 5 minutes & 5\% of all gateways in the same district \\
    \hline
    SAE monitor failure & NE\_e alarm, audit monitor alarm, NE threshold alarm & 5 minutes & 10 in the same NE \\
    \hline
    NE signal failure & Signalling link alarm & 3 minutes  & 10 in the same NE \\
    \hline
    Billing failure & Billing equipment fault, abnormal storage of billing data, billing file fault & 3 minutes & 10 in a district\\
    \hline
    Fiber link failure &  IP signal alarm, optical fiber fault & 3 minutes & 5 in a district\\
    \hline
    AC failure & DHCP resource alarm, AC power alarm, Portal service fault, Radius billing server fault, Radius authentication fault & 5 minutes & 5 in a district \\
    \bottomrule
    \end{tabular}%
  \label{tab:rules}%
\end{table*}%
\eat{
\begin{figure}[t]
\centering
\subfloat[][Synthetic]{\epsfxsize=0.4\columnwidth \epsffile{figure/mem_real.eps}}
\subfloat[][Telecom. alarms]{\epsfxsize=0.4\columnwidth \epsffile{figure/mem_syn.eps}}
\vspace{-2ex}\caption{Memory Consumption by varying $k$}\label{fig:exp333}
\vspace{-2ex}\end{figure}
}

%
\vspace{0ex}\section{Conclusion}\label{sec:con}
In this work, we present ONCE and ONCE+ algorithms, which can answer non-overlapped and distinct frequency counting problem respectively, for given time-constrained serial episodes within a given streaming sequence. ONCE and ONCE+ algorithms work in a one-pass way with the help of a carefully designed data structure, OccMap. For each single event arrived, ONCE and ONCE+ only takes $O(k\log\tau)$ time to process it. In fact, the problem we are addressing in this work can degenerate to the traditional serial episode frequency mining problem if time constraint for the target episode is set to infinite. Moreover, we theoretical prove that ONCE (\resp ONCE+) can correctly answer the non-overlapped (\resp distinct) frequency counting problem.   Experimental study conducted over both synthetic and real world datasets justify that ONCE and ONCE+ can efficiently work on stream data, even if millions of events arrive in a single second.

In some complex applications, there exist more complex serial episodes that consist of multi-dimensional signals. Extending ONCE and ONCE+ to address the same problem in high-dimensional streams is part of our future work.

\appendices


\vspace{0ex}\section{Proof of Lemma~\ref{lemma1}}\label{app1}
Suppose $freq'(e^\tau,S)> freq(e^\tau,S)$, that is, either one of the following cases happens.
\begin{itemize}
  \item ONCE finds a minimum occurrence $Occ_{OPT}(e^\tau,S)$ which does not pass time constraint test, but we incorrectly increase $freq'(e^\tau,S)$ by $1$.
  \item ONCE find two minimum occurrences $Occ_{OPT1}(e^\tau,S)$,\linebreak $Occ_{OPT2}(e^\tau,S)$ that overlap with each other and both pass time constraint test, we count the frequency of $e^\tau$ by $2$.
\end{itemize}

According to Algorithm~\ref{alg:Framwork} and~\ref{alg:valeli}, only when $Occ_{OPT}(e^\tau,S)$ passes the time constraint test, we increase its frequency. Therefore, the first case cannot happen in ONCE.

Then we prove that the second case contradicts the conditions in ONCE. Without loss of generality, suppose $Occ_{OPT1}(e^\tau,S)$ consists of $t_1,\ldots,t_k$ and $Occ_{OPT2}(e^\tau,S)$ consists of $t_1^\prime,\ldots,t_k^\prime$, respectively, and $t_k\le t_k^\prime$. If $t_1\ge t_1^\prime$, then according to Definition~\ref{df:minocc}, $Occ_{OPT2}(e^\tau,S)$ cannot be a minimum occurrence. Hence, $t_1<t_1^\prime$. As there is overlap between $Occ_{OPT1}(e^\tau,S)$ and $Occ_{OPT2}(e^\tau,S)$, then $\exists i\in[2,k]$. Thus, $t_i\ge t_1^\prime$.

According to Algorithm~\ref{alg:valeli}, when $Occ_{OPT1}(e^\tau,S)$ is found by ONCE and passes time constraint test, all the entries in $OM(e^\tau)$ are eliminated, including $t_1^\prime$. After that, any other entries $t$ inserted satisfies $t>t_k$.

Obviously, as long as $Occ_{OPT1}(e^\tau,S)$ are found and tested, $Occ_{OPT2}(e^\tau,S)$ can never be found by ONCE as $t_1^\prime$ has already been eliminated. Therefore, the second case can never happen in ONCE.

In all, all the cases that lead to $freq'(e^\tau,S)> freq(e^\tau,S)$ can never happen in ONCE, that is $freq'(e^\tau,S)\le freq(e^\tau,S)$.

\vspace{-1ex}\section{Proof of Lemma~\ref{lemma2}}\label{app2}
Suppose $freq'(e^\tau,S)< freq(e^\tau,S)$, that is, there exists a minimum occurrence $Occ_{OPT}(e^\tau,S)$ (\ie $[t_1,\ldots,t_k]$) that satisfies time constraint ($t_k-t_1\le\tau$) and does not overlap with any other minimum occurrences, which successfully increase the frequency, but ONCE fails to find it. In the following, we show that the premise cannot happen in ONCE through two steps. Firstly, we prove that when $t_k$ is appended into $L_k$, $\forall t<k, t_i\in OM(e^\tau)$. Secondly, we prove that as long as $t_1,\ldots,t_k$ are present in $OM(e^\tau)$, ONCE can never miss them.

\vspace{0ex}\btitle{Step 1.} Without loss of generality, suppose $t_i\not\in OM(e^\tau)$ and $t_j\in OM(e^\tau)$ for all $j<i$ ($t_i$ is the first one that is not present in $OM(e^\tau)$).

According to Algorithm~\ref{alg:valeli}, $t_i$ may either be eliminated from $OM(e^\tau)$ or not inserted in $OM(e^\tau)$.

\noindent \textbf{1)} [\emph{if eliminated}] According to Algorithm~\ref{alg:Framwork},~\ref{alg:update} and~\ref{alg:valeli}, the elimination only happen in Line 7,12,15 of Algorithm~\ref{alg:valeli} and Line 5 of Algorithm~\ref{alg:update}, we discuss each of the case in sequence.
\begin{itemize}
\item If it is eliminated from $OM(e^\tau)$ in Line 7 of Algorithm~\ref{alg:valeli}, then $[t_1,\ldots,t_k]$ overlaps with another minimum occurrences that successfully increases the frequency, which contradicts with the premise.

\item If it is eliminated from $OM(e^\tau)$ in Line 12 of Algorithm~\ref{alg:valeli}, all $t_j$ ($j<i$) will also be eliminated, this contradicts with the fact that $t_j$ are present in $OM(e^\tau)$; otherwise $i=1$, as it has been eliminated in Line 12, then $t_k-t_i = t_k-t_1 > L_k[1]-t_1 > \tau$, which contradicts with the fact that $t_k-t_1\le\tau$.

\item If it is eliminated from $OM(e^\tau)$ in Line 15 of Algorithm~\ref{alg:valeli}, then according to Line 14, we can derive $i\ge 2$. All $t_j$ ($j<i$) will also be eliminated as $t_i\le L_{i-1}[1]\le\ldots\le L_1[1]$, which contradicts with the fact that $t_j$ are present in $OM(e^\tau)$.

\item If it is eliminated from $OM(e^\tau)$ in Line 5 of Algorithm~\ref{alg:update}, it is easy to find $t_k-t_i>\tau$, therefore $t_k-t_1>\tau$, which contradicts with the fact that $t_k-t_1\le\tau$.

\end{itemize}

Therefore, $t_i$ can never be eliminated from $OM(e^\tau)$. Then it should never be inserted in to $OM(e^\tau)$.

\noindent \textbf{2)} [\emph{if not inserted}] If $t_i$ is never inserted into $OM(e^\tau)$, \ie $t_i$ is never appended to $L_i$, then $i>1$ and $L_{i-1}$ should be empty when the event $S[j]$ ($S[j].t=t_i$) arrives according to Algorithm~\ref{alg:update}. As $L_{i-1}$ is empty, $t_{i-1}$ should not be present in $OM(e^\tau,S)$, which contradicts with the premise that $t_i$ is the first one that is not present in $OM(e^\tau)$.

Taking both 1) and 2) together, $\forall i\in[1,k]$, $t_i$ should be present in $OM(e^\tau,S)$ when $t_k$ is appended to $L_k$.

\vspace{0ex}\btitle{Step 2.} As all $t_1,\ldots,t_k$ are present in $OM(e^\tau,S)$ when $t_k$ is appended to $L_k$, we further show that ONCE can never miss them.

Suppose ONCE algorithm finds a minimum occurrence $[t_1^\prime,\ldots,t_{k-1}^\prime,t_k]$ when $t_k$ is appended. If $t_1^\prime>t_1$, $[t_1,\ldots,t_k]$ is not a minimum occurrence, which contradicts with the premise of the lemma.

Otherwise, $t_1^\prime<t_1$, then $[t_1^\prime,\ldots,t_{k-1}^\prime,t_k]$ cannot be a minimum occurrence. Therefore, $t_1^\prime=t_1$. As $t_k-t_1^\prime = t_k-t_1 \le \tau$, ONCE will definitely increase $freq'(e^\tau,S)$ by $1$ and remove all entries in $OM(e^\tau)$.

Taking both step 1 and 2 together, $freq'(e^\tau,S)< freq(e^\tau,S)$ can never happens, thus $freq'(e^\tau,S)\ge freq(e^\tau,S)$.
\section*{Acknowledgment}

This work is supported by National Nature Science Foundation of China (No. 61672408, 61472298), CCF-VenustechRP (No. 2017005), Huawei Innovative Research Program (No. HIRPO20160606), China 111 Project (No. B16037).


\end{document}